\documentclass[10pt]{IEEEtran}
\usepackage{cite}
\usepackage{amsmath,amssymb,amsfonts,amsthm,mathtools}
\usepackage{graphicx}
\usepackage{placeins}
\usepackage{textcomp}
\usepackage{xcolor}
\usepackage{blindtext}
\usepackage{lipsum}
\usepackage{subfig}
\usepackage{cuted}
\usepackage{physics}
\usepackage[long]{optidef}
\usepackage[detect-none]{siunitx}
\usepackage{multirow}
\usepackage{array}
\usepackage{steinmetz}
\usepackage{bigints}
\newtheorem{theorem}{Theorem}
\newtheorem{lemma}[theorem]{Lemma}
\usepackage{epstopdf}
\usepackage{cuted}
\newcommand{\Hash}{\mathsf{Hash}}
\usepackage{algorithm}
\usepackage{algpseudocode}
\usepackage{caption}
\usepackage{subcaption}
\usepackage{soul}

\usepackage[labelfont=normalfont, labelsep=period, subrefformat=parens, labelformat=parens]{subcaption}

\usepackage{stmaryrd}
\usepackage{mathtools, tikz-cd}
\usetikzlibrary{cd}
\usepackage{mdframed}
\newmdtheoremenv{theo}{Theorem}

\begin{document}

\title{Engel $p$-adic Isogeny-based  Cryptography over Laurent Series: Foundations, Security, and an ESP32 Implementation}

\author{Ilias Cherkaoui~and~Indrakshi~Dey,~\IEEEmembership{Senior Member,~IEEE,}
\thanks{I.~Cherkaoui (Email: Ilias.Cherkaoui@WaltonInstitute.ie) and I. Dey (Email: indrakshi.dey@waltoninstitute.ie) are with the Walton Institute for Information and Communications Systems Science, South East Technological University, Waterford, Ireland.}
\thanks{This work is supported in part by HEARG TU RISE Project ``AIQ-Shield" and the HORIZON European Cybersecurity Competence Centre (ECCC) Project ``Q-FENCE" under Grant Number 101225708.}
}


\maketitle

\begin{abstract}
Securing the Internet of Things (IoT) against quantum attacks requires public-key cryptography that (i) remains \emph{compact} and (ii) runs \emph{efficiently on microcontrollers}—capabilities many post-quantum (PQ) schemes lack due to large keys and heavy arithmetic. We address both constraints simultaneously with, to our knowledge, the \textit{first-ever} isogeny framework that encodes super-singular elliptic-curve isogeny data via \textit{novel} Engel expansions over the $p$-adic Laurent series $\mathbb{Q}_p((t))$. Engel coefficients compress torsion information, thereby addressing the compactness constraint, yielding public keys of $\sim$1.1–16.9~kbits preserving the hallmark small sizes of isogeny systems. Engel arithmetic is local and admits fixed-precision $p$-adic operations, enabling micro-controller efficiency with low-memory, branch-regular kernels suitable for embedded targets.

Methodologically, we prove closure and isomorphism properties for Engel arithmetic; derive explicit 2-isogenies via Hensel lifts of torsion coordinates; and build a supersingular-isogeny Diffie–Hellman (SIDH)-style \emph{key encapsulation mechanism (KEM)} that publishes truncated Engel coefficients with $j$-invariants. An ESP32 implementation validates practicality for IoT: encrypt/decrypt latency scales linearly with message size and improves with clock frequency ($\approx$51\% from 80$\to$160~MHz; $\approx$34\% from 160$\to$240~MHz); power traces exhibit operation-specific signatures matching the algorithmic work model. Security is framed by the \emph{Supersingular Isogeny Decisional Problem (SIDP)} and a new \emph{Engel-Expansion Inversion} problem, whose non-periodic $p$-adic structure resists Shor-type quantum attacks. Overall, the proposed Engel/$p$-adic isogeny KEM directly delivers compact, microcontroller-ready PQ security for IoT nodes. \footnote{\copyright 2025 IEEE. Personal use of this material is permitted. Permission from IEEE must be obtained for all other uses, in any current or future media, including reprinting/republishing this material for advertising or promotional purposes, creating new collective works, for resale or redistribution to servers or lists, or reuse of any copyrighted component of this work in other works.}
\end{abstract}


\section{Introduction}
\IEEEPARstart{T}{he} advent of quantum computing has profoundly transformed the landscape of secure communication and data protection. With Shor’s and Grover’s algorithms threatening the classical hardness assumptions underlying RSA, ElGamal, and elliptic-curve cryptography (ECC), the world faces an urgent need for post-quantum cryptography (PQC), cryptosystems designed to resist attacks from both classical and quantum computers. Among the diverse PQC candidates, {isogeny-based cryptography} stands out as one of the most mathematically sophisticated and promising approaches. Its strength derives from the computational hardness of finding isogenies (algebraic morphisms) between supersingular elliptic curves, a problem for which no efficient quantum algorithm exists.

However, despite their theoretical elegance and compact key sizes, isogeny-based systems are notoriously resource-intensive. The arithmetic operations required—modular multiplications, inversions, and isogeny evaluations over large finite fields—introduce substantial latency and energy overhead. Consequently, the deployment of such systems on constrained Internet of Things (IoT) devices has remained an open challenge. De Feo and Jao’s introduction of the Supersingular Isogeny Diffie–Hellman (SIDH) protocol, later extended to SIKE (Supersingular Isogeny Key Encapsulation) \cite{feo,nis}, demonstrated the practical potential of isogeny-based security. Yet, even with advances by Costello et al. on ARM microcontrollers \cite{cos} and by Koziel et al. on FPGA platforms \cite{koz}, the arithmetic density of these systems continues to limit their applicability to ultra-low-power hardware.

{This paper presents, for the first time, a fundamentally new cryptographic construction that integrates Engel expansions with \(p\)-adic Laurent-series arithmetic to realize an isogeny-based cryptosystem optimized for embedded systems.} This integration bridges two domains, analytic number theory and post-quantum cryptography, that have historically evolved in isolation. Engel expansions, originally developed in Diophantine approximation \cite{sch}, provide a systematic representation of numbers as convergent reciprocal sums of monotonically increasing integer sequences. When transplanted into \(p\)-adic number systems and formal Laurent-series fields, these expansions yield compact, deterministic, and algebraically regular arithmetic forms that can encode elliptic-curve parameters and torsion structures with minimal redundancy.  

The resulting {Engel–Laurent framework} introduces a structural regularity that allows predictable and constant-time arithmetic, bounded memory usage, and stable convergence in isogeny evaluations. This property is particularly beneficial for embedded cryptographic devices, where side-channel resistance, energy efficiency, and timing determinism are paramount. By embedding the Engel expansion into the \(p\)-adic topology, this work provides a new parameterization method that stabilizes arithmetic flows and compresses data representation in supersingular isogeny graphs. 

The original Supersingular Isogeny Diffie--Hellman (SIDH) protocol and its KEM variant SIKE were recently shown to be vulnerable to a powerful key-recovery attack by Castryck and Decru~\cite{SIKEp}. The attack exploits the publication of auxiliary torsion images together with a specific multi-prime isogeny structure over the finite field $\mathbb{F}_{p^{2}}$. Our construction borrows only the high-level ``SIDH-style'' KEM interface (a shared $j$-invariant), but differs substantially at the algebraic level: it is defined over the characteristic-zero field $\mathbb{Q}_p((t))$, does not publish images of independent torsion bases under the secret isogeny, and relies on Engel-encoded $p$-adic Laurent series rather than affine torsion points. Consequently, the current SIDH and SIKE attacks do not transfer directly to our scheme. Nevertheless, as with any new isogeny-based primitive, our security claims are conditional, and a dedicated cryptanalysis of this $p$-adic Engel setting remains an important direction for future work.

The proposed system is implemented on the ESP32 microcontroller — a dual-core, WiFi-enabled embedded platform widely adopted in IoT deployments. The experiments demonstrate a full realization of post-quantum encryption, decryption, and key generation within the strict limits of embedded memory and processing capacity. Results reveal stable execution latency, reduced power consumption, and low timing variance, thereby validating that Engel–Laurent arithmetic not only enhances theoretical security but also maps effectively onto physical constraints.

\subsection{Contribution}
The primary contributions of this work are multifold and define the first-ever integration of Engel expansions with isogeny-based cryptography in a \(p\)-adic Laurent-series framework:
\begin{itemize}
    \item We introduce a novel algebraic representation that combines Engel expansions with \(p\)-adic Laurent-series, providing a compact and deterministic encoding of supersingular elliptic-curve parameters for isogeny-based cryptography.
    \item We demonstrate the first implementation of an Engel–Laurent parameterized isogeny-based cryptosystem on a constrained hardware platform (ESP32), achieving efficient key generation, encryption, and decryption in real time under limited computational and energy budgets.
    \item We conduct detailed quantitative analyses of computational latency, energy consumption, and power trace behavior, providing a direct physical interpretation of the underlying arithmetic operations.
    \item We show that the Engel–Laurent arithmetic naturally supports constant-time execution, thereby mitigating side-channel risks and ensuring operational determinism at the microcontroller level.
\end{itemize}

By merging theoretical constructs from number theory and practical concerns from embedded systems design, this work establishes a new paradigm for lightweight post-quantum cryptography that is simultaneously secure, compact, and physically realizable.

The remainder of this paper is structured as follows. Section~\ref{Sec: 2_system_model_problem} introduces the system model and problem formulation for the embedded cryptographic implementation. Section~\ref{Sec: 3_Auction_based_allocation} presents the mathematical foundations of the proposed Engel–Laurent representation and its relation to isogeny-based arithmetic. Section~\ref{Sec: 5_proposed_approach_mSAA} details the embedded implementation on the ESP32 platform, including architectural optimization and algorithmic scheduling. Section~\ref{Sec: 6_results_discussion} provides a comprehensive presentation of the numerical results and physical interpretation of the findings. Finally, Section~\ref{Sec: 7_conclusion} concludes the paper with key insights and perspectives for future research.  

\subsection{Related Works}

The foundations of isogeny-based cryptography were established by De Feo and Jao through the SIDH protocol, later expanded into SIKE \cite{feo,nis}. These systems introduced compact, post-quantum-secure key exchange mechanisms built upon the hardness of computing isogenies between supersingular elliptic curves. Subsequent works by Costello et al. \cite{cos} and Koziel et al. \cite{koz} focused on optimizing isogeny evaluation and field arithmetic for embedded processors and FPGA architectures. Despite these advances, performance and power constraints remain key obstacles to their practical use in IoT networks.

Parallel to these developments, number theorists have studied Engel expansions and their \(p\)-adic analogues extensively for over a century. Schuske and Koblitz \cite{sch,kob} explored their structural and arithmetic properties, while Kedlaya \cite{ked} demonstrated their potential applications in formal power series and coding theory. However, their incorporation into cryptographic arithmetic, especially as a means to encode elliptic-curve parameters, has never been attempted before. The present work thus constitutes the first instance where Engel expansions and \(p\)-adic Laurent-series are operationally merged with isogeny-based cryptography.

Lightweight cryptography for IoT platforms has generally relied on symmetric systems and ECC optimizations, such as TinyECC \cite{liu}, which focuses on improving classical ECC performance in constrained environments. Lattice-based post-quantum systems have also been explored in this context, such as the work of Oder et al. \cite{od}, yet they often incur substantial increases in code size, memory footprint, and power demand. Isogeny-based cryptography, by contrast, offers compact key structures that naturally suit communication-constrained IoT systems, but it demands efficient arithmetic formulations, a challenge addressed in this paper through Engel–Laurent representations.  

In summary, previous studies have either enhanced classical ECC for embedded systems or sought to miniaturize lattice-based PQC. None have reformulated the arithmetic foundations of isogeny-based systems for embedded operation. The present work therefore marks the first convergence of Engel expansions, \(p\)-adic arithmetic, and isogeny-based cryptography within an experimentally validated embedded implementation, a novel path bridging pure number theory and practical post-quantum engineering.

\section{Field Construction}\label{Sec: 2_system_model_problem}

We build the $p$-adic/Laurent-series setting and the Engel-expansion machinery that underpins our isogeny construction. To keep the flow explanatory and actionable, all formal proofs are provided in the Appendices with cross-references here. 

\subsection{Mathematical Preliminaries}
\subsubsection{Preliminaries on $p$-adics and Valuation}\label{subsec:padics}
Let $p$ be prime. The field $\mathbb{Q}_p$ is the completion of $\mathbb{Q}$ under the $p$-adic absolute value $|\cdot|_p$. Any $x\in \mathbb{Q}_p$ admits an expansion
\begin{align}
x=\sum_{n=-k}^{\infty} a_n p^n,\qquad a_n\in\{0,1,\dots,p-1\},\ k\in \mathbb{Z},
\end{align}
with $p$-adic valuation $v_p(x)$ defined as the smallest integer $n$ for which $a_n\neq 0$; set $v_p(0)=+\infty$. The absolute value is $|x|_p=p^{-v_p(x)}$. In code, $v_p$ measures the number of trailing $p$-adic “zeros”—directly informing fixed-precision limb allocation and error budgeting; larger $v_p$ means smaller magnitude in the $p$-adic sense, enabling early truncation without destabilizing arithmetic.

\subsubsection{Laurent Series over $\mathbb{Q}_p$}\label{subsec:laurent}
We work over the Laurent-series field
\begin{align}
\mathbb{Q}_p((t))\triangleq \left\{ f(t)=\sum_{n=-k}^{\infty} a_n t^n  \Big|\ a_n\in\mathbb{Q}_p,\ k\in \mathbb{N} \right\}.
\end{align}
The topology combines the $p$-adic valuation on coefficients with the $t$-adic order on exponents. We use coefficients in $\mathbb{Z}_p[[t]]$ when enforcing integrality and convergence properties. The $t$-variable tracks formal series structure (e.g., parameter drift or lifted roots), while $p$ controls arithmetic stability. On micro-controllers, we represent $f(t)$ by a bounded window of $t$-exponents and a bounded number of $p$-adic limbs per coefficient.

\subsubsection{Engel Expansions on $\mathbb{Q}_p((t))$}\label{subsec:engel}
For $f(t)\in\mathbb{Q}_p((t))$, an \emph{Engel expansion} is a convergent series
\begin{align}\label{eq:engel-def}
f(t)&=\sum_{n=1}^\infty {1}/{a_1(t)a_2(t)\cdots a_n(t)}, \nonumber\\
& a_i(t)\in \mathbb{Z}_p[[t]],\ a_i(0)\equiv 1\pmod p,
\end{align}
constructed greedily by
\begin{align}\label{eq:engel-greedy}
a_1(t)=\big\lfloor 1/f(t)\big\rfloor_{p};  
a_{k+1}(t)=\left\lfloor {1}/{\,f(t)-\sum_{i=1}^k \frac{1}{\prod_{j=1}^i a_j(t)}\,}\right\rfloor_{p},
\end{align}
where $\lfloor\cdot\rfloor_{p}$ extracts the $p$-adic leading term. Convergence follows from
\begin{align}\label{eq:engel-vp-growth}
v_p\!\left({1}/{\prod_{j=1}^n a_j(t)}\right)\ \ge\ n~\rightarrow~\lim_{n\to\infty}v_p\!\left({1}/{\prod_{j=1}^n a_j(t)}\right)=+\infty,
\end{align}
and uniqueness is enforced by the constraints $a_i(t)\in \mathbb{Z}_p[[t]]$ and $a_i(0)\equiv 1\pmod p$ (see Appendix A). Engel coefficients act like a {compressed, stable} coordinate system for $f(t)$: each new factor increases $p$-adic valuation by at least $1$, guaranteeing rapidly diminishing residual terms. Practically, storing a short prefix of $a_i(t)$ captures $f(t)$ to the needed precision with predictable truncation error.

\subsection{Engel Field Isomorphism and 2‑Isogeny via V\'elu}
\subsubsection{The Engel--Field Correspondence}\label{subsec:engel-field}
Let $\mathcal{E}$ be the set of series in the form \eqref{eq:engel-def} produced by \eqref{eq:engel-greedy}.

\begin{theorem}[Engel--Field Isomorphism]\label{thm:engel-isom}
The map $\phi:\mathbb{Q}_p((t))\to\mathcal{E}$ sending $f(t)$ to its Engel expansion is bijective and preserves addition, multiplication, and inversion (on $\mathbb{Q}_p((t))^\times$). Consequently $(\mathcal{E},+,\cdot)$ is a field isomorphic to $\mathbb{Q}_p((t))$.
\end{theorem}

Engel arithmetic is \emph{structure-preserving}: one can compute in Engel space (add/multiply/invert) or in $\mathbb{Q}_p((t))$ and obtain consistent answers. This underlies our ``encode once, compute locally” approach in the KEM. For Proof, see Appendix B.

\subsubsection{Elliptic Curves over $\mathbb{Q}_p((t))$ and Reduction}\label{subsec:elliptic}
Consider the short Weierstrass model
\begin{align}\label{eq:E-base}
E:\ y^2=x^3+A(t)x+B(t),~A(t),B(t)\in \mathbb{Q}_p((t)),
\end{align}
with nonzero discriminant $\Delta(t)=-16(4A(t)^3+27B(t)^2)\neq 0$. If $A(t),B(t)\in Z_p[[t]]$, reduction modulo $p$ yields
\begin{align}
\widetilde{E}/\mathbb{F}_p((t)):\ y^2=x^3+\widetilde{A}(t)x+\widetilde{B}(t),
\end{align}
and for super-singular $\widetilde{E}$, $\mathcal{E}_{\mathbb{F}_p((t))}(\widetilde{E})\otimes\mathbb{Q}$ is a quaternion algebra. Working integrally ensures good reduction and enables us to port structural properties (e.g., super-singularity) across isogenies and reductions—key for security modeling and parameter selection.

\subsubsection{Hensel Lifts of $2$-Torsion and Engel Encoding}\label{subsec:hensel}
Fix $p>3$ and the curve
\begin{align}\label{eq:base-curve}
E:\ y^2=x^3+1+pt\quad \text{over }\mathbb{Q}_p((t)).
\end{align}
Setting $y=0$ requires $x^3=-(1+pt)$. If $p\equiv 1\pmod 3$, a primitive cube root $\zeta_3\in\mathbb{F}_p^\times$ exists. Let $f(X)=X^3+(1+pt)\in Z_p[[t]][X]$ and $\alpha_0=\zeta_3$; then $f(\alpha_0)\equiv 0\pmod p$ and $f'(\alpha_0)=3\zeta_3^2\not\equiv 0\pmod p$. By Hensel’s lemma, there is a unique lift $\alpha=x(t)\in\mathbb{Q}_p((t))$ with
\begin{align}\label{eq:xt-series}
x(t)&=\zeta_3(1+pt)^{1/3}=\zeta_3\Big(1+\tfrac{pt}{3}-\tfrac{(pt)^2}{9}+\cdots\Big),\nonumber\\
&v_p\!\left(\frac{(pt)^n}{3^n}\right)\ge n.
\end{align}
Then $P=(x(t),0)$ is a $2$-torsion point and $x(t)$ admits a compact Engel encoding. Hensel lifting provides a {deterministic, convergent} way to compute torsion coordinates. Encoding $x(t)$ via Engel coefficients compresses this root and makes subsequent isogeny evaluations memory-light.

\subsubsection{Explicit $2$-Isogeny via V\'{e}lu}\label{subsec:velu}
Let $G=\langle P\rangle$ with $P=(\alpha,0)$, $\alpha=x(t)$. V\'{e}lu’s formula specializes to
\begin{align}\label{eq:velu-map}
\phi(x,y) &=\left(\,x+\frac{c}{x-\alpha},~y\left(1-\frac{c}{(x-\alpha)^2}\right)\right),\nonumber\\ 
c &=3\alpha^2+A(t)
\end{align}
with codomain $E'=E/G$ in short Weierstrass form $y^2=x^3+A'x+B'$, where $A'=A-5c$, $B'=B-7c\,\alpha$. For \eqref{eq:base-curve}, $A(t)=0$, $B(t)=1+pt$, whence $c=3x(t)^2$ and $A'=-15\,x(t)^2$; $B'=(1+pt)-21\,x(t)^3$. Using $x(t)^3=-(1+pt)$ we obtain
\begin{align}\label{eq:Eprime}
E':\ y^2=x^3-15\,x(t)^2\,x+\big(22+22\,pt\big).
\end{align}
\noindent The map depends only on $\alpha=x(t)$ and $c=3\alpha^2$—both Engel-encodable. Computation mainly uses differences $(x-\alpha)$, which are numerically stable in the $p$-adic sense and translate to branch-regular kernels.

\subsubsection{Reduction and Super-singularity Preservation}\label{subsec:reduction}
Reducing \eqref{eq:Eprime} modulo $p$ gives
\begin{align}
\widetilde{E}':\ y^2=x^3-15\,\widetilde{x(t)}^2\,x+22\quad\text{over }\mathbb{F}_p((t)).
\end{align}
Under the standard hypotheses (char$(\mathbb{F}_p)\neq 2,3$), degree-$2$ isogenies induce isomorphisms on $\mathbb{Q}$-endomorphism algebras; thus supersingularity is preserved (see Appendix).

\begin{theorem}[Supersingularity via $2$-Isogeny]\label{thm:SS-preserve}
Let $E/\mathbb{Q}_p((t))$ be as in \eqref{eq:base-curve}, $P=(x(t),0)$, and $\phi:E\to E'=E/\langle P\rangle$ the isogeny \eqref{eq:velu-map}. Then the induced map
\[
\phi_*:\mathcal{E}(E)\otimes\mathbb{Q} \longrightarrow \mathcal{E}(E')\otimes\mathbb{Q},\qquad \psi\mapsto \phi\circ\psi\circ\hat{\phi},
\]
is an isomorphism. Consequently, if $\widetilde{E}$ is supersingular over $\mathbb{F}_p((t))$, then so is $\widetilde{E}'$. 
\end{theorem}
\noindent Security-wise, the endomorphism algebra (quaternion) structure that underlies isogeny hardness is preserved along our map, validating that Engel-encoded evaluations do not weaken the super-singular setting. For Proof, see Appendix B. The constructions above yield two actionable levers:
\begin{enumerate}
\item \textbf{Compact encoding:} Engel series of $x(t)$ and $c=3x(t)^2$ provide short, high-valuation descriptors of torsion data, minimizing public key material.
\item \textbf{Efficient evaluation:} The rational form \eqref{eq:velu-map} requires only differences and a handful of multiplications/divisions in fixed $p$-adic precision, enabling constant-memory, branch-regular kernels on ESP32-class devices.
\end{enumerate}
These properties directly address the compactness and microcontroller-efficiency challenges laid out in the Introduction. All proofs referenced in this section are presented in Appendices A-C (supplement file).

\section{Algorithm Design}\label{Sec: 3_Auction_based_allocation}

This section specifies an Engel/$p$-adic, supersingular-isogeny encryption scheme following the \emph{key encapsulation mechanism (KEM)} pattern used in isogeny-based designs (like SIKE). 

\subsection{Design Goals, Setting, and Assumptions}\label{subsec:goals}

Our overarching goals are centered on developing a cryptographic scheme that is both efficient for embedded systems and resilient against quantum threats. To achieve the objective of \textbf{compactness}, the design minimizes the required public data by ingeniously encoding torsion point coordinates via \textbf{Engel expansions} over the field of formal Laurent series with $p$-adic coefficients, denoted as $\mathbb{Q}_p((t))$. This specialized encoding significantly reduces the size of the public parameters. For optimized \textbf{embedded efficiency}, the scheme is meticulously structured to rely exclusively on \textbf{fixed-precision $p$-adic arithmetic} implemented with highly structured, \textbf{branch-regular kernels}. This fixed-precision methodology is vital for practical hardware implementation on resource-constrained devices, as it ensures predictable execution timing, thereby mitigating the risk of timing-based side-channel attacks. Finally, Post-quantum resilience in our construction is modeled via two conjecturally hard problems, namely the SIDP and EEIP. These assumptions are intended to provide a conservative security foundation against quantum adversaries, subject to future cryptanalytic scrutiny of the proposed $p$-adic Engel isogeny framework.

\subsubsection{Field/curve setting} 

Let $p>3$ be prime (we later consider $p\equiv 2\pmod 3$ to keep reduction super-singular while allowing Hensel lifts over $\mathbb{Q}_p((t))$). Let us Work over $E/\mathbb{Q}_p((t)): y^2=x^3+1+pt$, and fix the $2$-torsion point $P=(x_P(t),0)$ with,
\begin{align}
x_P(t)=\zeta_3(1+pt)^{1/3}\in\mathbb{Q}_p((t))\qquad\text{(Hensel lift)}, 
\end{align}
where $\zeta_3$ is a primitive cube root of unity in an algebraic extension (its reduction may or may not lie in $\mathbb{F}_p$, cf. Sec.~II). Denoted by $\phi:E\to E'$, $\deg\phi=2$, the V\'{e}lu isogeny with kernel $\langle P\rangle$ as in \eqref{eq:velu-map}--(12) (Sec.~II). We encode $x_P(t)$ and derived quantities via truncated Engel series with coefficients in $Z_p[[t]]$. The single $2$-isogeny primitive is deliberately small and uniform. All public data is ultimately functions of a few $p$-adic series (Engel coefficients), which are cheap to store and manipulate on ESP32-class devices.

\subsubsection{High-Level KEM Pattern}\label{subsec:pattern}
We adopt a SIDH/SIKE-style KEM interface: a) \textsf{KeyGen}$\to(\textsf{pk},\textsf{sk})$, b) \textsf{Encap}$(\textsf{pk})\to (c,K)$, and c) \textsf{Decap}$(\textsf{sk},c)\to K$. Correctness is via a shared isogeny descendant $E^{(n+r)}$ and $K=\#(j(E^{(n+r)}))$; see Appendix D for the formal proof. The shared key is derived from a curve identifier ($j$-invariant), avoiding large payloads. Hashing $j(\cdot)$ aligns with standard isogeny KEM design and simplifies constant-time coding.

\subsubsection{Key Material and Public Encodings}\label{subsec:keymaterial}
\textbf{Private Key} - Choose $n\in\{1,2\}$ indicating the number of chained $2$-isogenies along the kernel $\langle P\rangle$: $\phi_n \;=\; \underbrace{\phi\circ\phi\circ\cdots\circ \phi}_{n\ \text{times}},\qquad E^{(n)} \;=\; E/\langle P\rangle^{n}.$
\textbf{Public Key} - Publish $\textsf{pk}=\Big(\{a_i^{(n)}(t)\}_{i=1}^{M},\; j(E^{(n)})\Big),$ where $\{a_i^{(n)}(t)\}$ are the Engel coefficients (to depth $M$ and coefficient depth $d$) encoding the torsion coordinate(s) and the induced V\'{e}lu outputs that define $E^{(n)}$. The secret is just a tiny counter ($n$); all heavy data are compressed into Engel coefficients plus a single scalar $j$-value. This minimizes storage and I/O.

\subsection{KEM Algorithms}\label{subsec:algorithms}

The following subsections detail the three essential algorithms — Key Generation, Encapsulation, and Decapsulation - that constitute the proposed Key Encapsulation Mechanism (KEM) based on isogenies and Engel expansions.

\subsubsection{Key Generation}\label{subsubsec:keygen}
The key generation process begins by sampling the \textbf{private key}, $n$, from the set $\{1, 2\}$, which dictates the number of chained $2$-isogenies to be computed along the kernel $\langle P\rangle$. The next step involves determining the resulting curve, $E^{(n)}$, by iteratively applying the Vélù formula $n$ times. To ensure computational stability during this isogeny evaluation, a truncated Engel series, $x_P^{(M)}(t)$, is employed in the Vélù transformation: 
\begin{align}
    \phi(x,y)=\left(x+\frac{3x_P^{(M)}(t)^2}{x-x_P^{(M)}(t)},\; y\left(1-\frac{3x_P^{(M)}(t)^2}{(x-x_P^{(M)}(t))^2}\right)\right).
\end{align}
Once the final curve $E^{(n)}$ is established, its defining coefficients $A_n(t)$ and $B_n(t)$ are computed, which then allows for the calculation of the curve's \textbf{$j$-invariant}: $j(E^{(n)})=1728\cdot\frac{4A_n(t)^3}{4A_n(t)^3+27B_n(t)^2}$. Finally, the \textbf{public key, $\textsf{pk}$} is output as the combination of the calculated $j$-invariant and the set of Engel coefficients, $\left\{a_i^{(n)}(t)\right\}_{i=1}^{M}$, while the secret key remains the simple counter $n$. From a hardware perspective, this entire process is highly efficient: each Vélù step simplifies down to a small number of fixed-size multiply and divide operations on $p$-adic limbs, and the pre-defined Engel truncation depth $M$ ensures the runtime remains deterministic, which is crucial for secure embedded implementations.

\subsubsection{Encapsulation}\label{subsubsec:encap}

The encapsulation process starts by sampling a short random integer, $r$, from $\{1, 2\}$, and computing the resultant curve $E^{(r)}=E/\langle P\rangle^{r}$ by performing $r$ chained Vélù isogenies from the base curve $E$. The core of the scheme involves computing the shared descendant curve, $E^{(n+r)}$, which can be efficiently computed from the base curve $E$ or by composing isogenies starting from either the public curve $E^{(n)}$ or the ephemeral curve $E^{(r)}$ (Appendix D). The shared secret key $K$ is then derived by hashing the $j$-invariant of this shared descendant curve, $K=\#(j(E^{(n+r)}))$. Finally, the ciphertext $c$ is output as the combination of the ephemeral curve $E^{(r)}$ and the masked key payload, $c=(E^{(r)}, \textsf{mask}\oplus K)$, or simply $c=(E^{(r)}, \textsf{Encaps}=K)$ in a Public Key Encryption (PKE) or Key Encapsulation Mechanism (KEM) context, respectively. Crucially, the total cost of encapsulation is highly optimized, requiring only one or two Vélù steps plus a hash operation, thereby entirely avoiding resource-intensive large-number exponentiation common in other cryptosystems.

\subsubsection{Decapsulation}\label{subsubsec:decap}

The decapsulation procedure is highly efficient and deterministic, beginning with the receipt of the ciphertext $c=(E^{(r)}, \cdot)$ and the knowledge of the simple private key $\textsf{sk}=n$. The recipient first computes the shared descendant curve, $E^{(n+r)}$, by applying $n$ chained $2$-isogenies starting from the ephemeral curve $E^{(r)}$ received in the ciphertext. Since the isogeny computed by the sender (length $r$) and the isogeny computed by the receiver (length $n$) commute, the receiver correctly arrives at the same final curve $E^{(n+r)}$. The shared secret key $K$ is then consistently derived by hashing the $j$-invariant of this final curve, $K=\#(j(E^{(n+r)}))$, which allows the recipient to successfully unmask and recover the payload (or return the derived key in a KEM). Decapsulation mirrors encapsulation: same arithmetic kernel, same constant-time hash. Because chaining $2$-isogenies along $\langle P\rangle$ produces the same descendant curve regardless of whether the party starts from $E$ or from the other party’s intermediate curve, both sides compute isomorphic curves with identical $j$-invariants (Appendix D). Thus both derive the same $K=\#(j(E^{(n+r)}))$. The scheme behaves like a two-step ladder: each side climbs $n$ or $r$ rungs, meeting at rung $n+r$.

\subsection{Security Rationale and Problem Statements}\label{subsec:security}

The security of the proposed cryptographic scheme is built upon the conjectured hardness of two distinct problems: the established Supersingular Isogeny Decisional Problem (SIDP), which provides quantum resistance, and the novel Engel-Expansion Inversion Problem (EEIP), derived from our specific encoding method.

\subsubsection{Supersingular Isogeny Decisional Problem (SIDP)}\label{subsubsec:SIDP}
\noindent
SIDP is defined over the set of supersingular curves, $\mathcal{E}$, typically over a suitable base field like $\mathbb{F}_{p^2}$, and involves two co-prime primes, $\ell_A$ and $\ell_B$. For uniformly random subgroups $\langle A\rangle\le E[\ell_A^m]$ and $\langle B\rangle\le E[\ell_B^n]$, the core challenge of SIDP is to distinguish the tuple $\big(E,E_A,E_B,\; j(E/\langle A,B\rangle)\big)$ from a randomly chosen tuple $\big(E,E_A,E_B,\; j(\mathcal{E}_{\text{rand}})\big)$, where $E_A=E/\langle A\rangle$, $E_B=E/\langle B\rangle$, and $\mathcal{E}_{\text{rand}}$ is a curve chosen uniformly from $\mathcal{E}$. It is required that for all Probabilistic Polynomial Time (PPT), even quantum, adversaries $\mathbb{A}$, the advantage in making this distinction must be negligible in the security parameter $\lambda$, with a formal definition provided in Appendix E.
\begin{figure*}[t]
    \centering
    \includegraphics[width=0.85\textwidth]{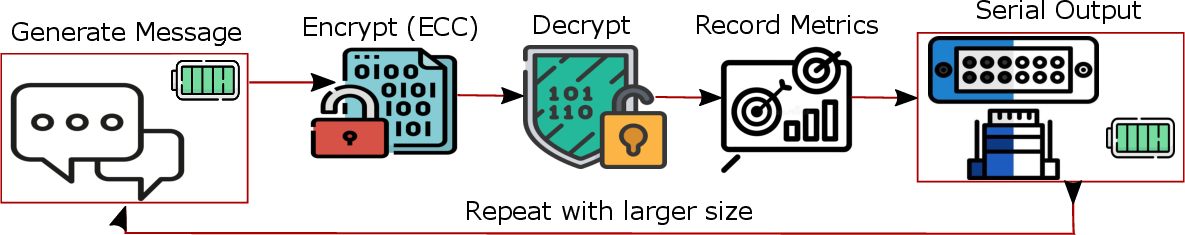}
    \caption{{{Single-Board Cryptographic Process Flow}: Sequential encryption–decryption and metric collection process on an ESP32 node.}}
    \vspace{-3mm}
    \label{fl2}
\end{figure*}

\subsubsection{Engel-Expansion Inversion Problem}\label{subsubsec:EEI}
\noindent
EEIP is based on the series $f(t)=\sum_{m\ge 1}\frac{1}{\prod_{i=1}^{m} a_i(t)}$, where the coefficients $a_i(t)$ are restricted to the $p$-adic integers in $t$, $\mathbb{Z}_p[[t]]$, and satisfy the condition $a_i(0)\equiv 1\pmod p$. Given only a truncation of this series, $f^{(M)}(t)$, where the depth $M$ is polynomial in the security parameter $\lambda$ ($M=\text{poly}(\lambda)$), the adversary's task is to recover the first $k=\text{poly}(\lambda)$ coefficients, specifically $a_1(t),\ldots,a_k(t)$. The difficulty of this problem arises because the underlying greedy recursion used to compute the series nonlinearly couples the coefficients in the $p$-adic topology, and crucially, there is no apparent reduction to a period-finding algorithm, which would otherwise make it vulnerable to Shor's algorithm (a formalization of this problem is detailed in Appendix E. 

EEIP is introduced here as a new, conjecturally hard inversion problem on truncated Engel expansions of $p$-adic Laurent series. As with other post-quantum assumptions (like, lattice or code-based problems), we do not prove unconditional hardness; rather, we use EEIP as an explicit cryptographic assumption underpinning our security reductions. Our analysis of EEIP focuses on showing that the sequence of Engel coefficients arising from our $p$-adic Laurent series does not exhibit a nontrivial finite period. This rules out a direct reduction of EEIP to standard QFT-based algorithms for period-finding and related hidden-subgroup problems. However, the absence of a short period does not in itself constitute a proof of hardness. In principle, novel quantum or classical cryptanalytic techniques could exploit additional structure, such as the algebraic properties of the $p$-adic field $\mathbb{Q}_p((t))$ or the greedy recursion defining the Engel coefficients. A thorough exploration of such potential attack strategies lies beyond the scope of this work and remains an important direction for future research.

\subsubsection*{Comparison with SIKE}

The Castryck--Decru attack on SIDH/SIKE targets a setting over $\mathbb{F}_{p^{2}}$ in which the public key exposes the images of torsion bases under a secret isogeny of composite degree. In contrast, our KEM publishes only truncated Engel coefficients $\{a_i^{(n)}(t)\}_{i=1}^{M}$ together with the curve invariant $j(E^{(n)})$, and the secret is restricted to the isogeny length $n$; there are no auxiliary torsion images analogous to those in SIKE. Combined with the different base field $\mathbb{Q}_p((t))$ and the non-periodic Engel recursion on $p$-adic coefficients, the algebraic structure exploited in the known SIDH and SIKE attacks is absent here. Our IND-CPA proof therefore assumes only the SIDP together with the hardness of EEIP, and does not reuse the now-invalidated computational assumptions of SIDH and SIKE.

\subsection{Parameters, Sizes, and Tuning}\label{subsec:params}

The scheme's balance of performance and security is controlled by four tunable parameters, or "knobs": (i) the prime bit-length $\lambda \in [8, 16]$ for $p$, which dictates storage complexity as $\mathcal{O}(\lambda)$ per $p$-adic limb; (ii) the Engel coefficient depth $d \in [4, 16]$, representing the number of coefficients per $a_i(t)$ and leading to storage of $\mathcal{O}(d\lambda)$ per $a_i$; (iii) the Engel series length $M \in [4, 32]$, which specifies the number of publicly stored coefficients and determines the total Engel-related storage as $\mathcal{O}(Md\lambda)$; and (iv) the isogeny degree $\log_2 \ell^e \in [2^{10}, 2^{16}]$, which is the dominant factor in the overall public key size. The private key remains extremely compact at $\lceil\log_2 2\rceil = 2$ bits. The total public key size is defined by the formula $|\textsf{pk}| = M d \lambda + \log_2 \ell^e$ bits. This yields predictable size examples: an IoT-level test configuration, where $(\lambda, d, M, \log_2\ell^e) = (8, 4, 4, 1024)$, results in a 1,152-bit public key; a 128-bit target security level, using $(8, 8, 8, 16384)$, yields 16,896 bits; and a high-security setting, $(16, 16, 16, 32768)$, requires 36,864 bits. This design reflects the physical interpretation that memory usage and throughput scale linearly with the Engel parameters $(M, d, \lambda)$, while the primary margin for cryptographic security is efficiently carried by the large value of the isogeny degree $\log_2 \ell^e$.

\subsection{Complexity and Comparison}\label{subsec:complexity}

The scheme's efficiency is highly competitive, assuming a message length $n$. Key generation costs $\mathcal{O}(Md+\ell)$ (driven by Engel builds and Vélù operations). Encapsulation and decapsulation costs are fast at $\mathcal{O}(\ell+n)$, primarily due to isogeny steps, plus a linear-time mask/XOR operation and constant-time hashing. For comparison, RSA with a $k$-bit modulus requires $\mathcal{O}(k^3)$ for key generation and $\mathcal{O}(k^2)$/$\mathcal{O}(k^3)$ for enc/dec, while Elliptic Curve Cryptography (ECC) with $b$-bit order costs $\mathcal{O}(b)$ for scalar multiplication but remains quantum-vulnerable. Crucially, our cryptographic kernels are dominated by a fixed pattern of multiplies and divides in $p$-adics, which is highly advantageous for secure embedded systems as it promotes deterministic timing and ease of masking. The security of the scheme is IND-CPA under the combined hardness of the Supersingular Isogeny Decisional Problem (SIDP, which yields a pseudo-random $j$-invariant) and the Engel-Expansion Inversion problem (formalized in an appendix). Furthermore, an upgrade to IND-CCA security can be achieved by applying the standard Fujisaki–Okamoto transform without requiring any alteration to the core, low-level arithmetic kernels. This FO wrapper essentially adds only meta-operations, such as hashing and re-encapsulation checks, around the same compact, arithmetic core. The field and Engel foundations provide the necessary arithmetic building blocks, and this analysis demonstrates how they are formed into a complete Key Encapsulation Mechanism (KEM). Correctness and final security proofs are provided in dedicated appendices.

\section{Methodology and Implementation}\label{Sec: 5_proposed_approach_mSAA}

Our experimental framework is designed to evaluate the computational and physical characteristics of Engel-based elliptic curve cryptography (ECC) on resource-constrained embedded hardware. Four ESP32 boards are configured in a chained network over WiFi, with each board performing encryption, decryption, and data transmission in sequence. A host computer collected timing and current data through UART serial communication. Static IP addressing and periodic UDP heartbeats maintained synchronization across the nodes, while retry mechanisms ensured packet integrity in the presence of wireless interference. CPU frequency is varied between $80$, $160$, and $240$~MHz to study how computational scaling interacts with network latency. The objective of this framework is to observe how theoretical linear-time arithmetic manifests under real hardware constraints and to interpret the observed deviations in terms of physical system behavior.

\subsection{Single-Board Implementation}

In the single-board configuration, one ESP32 node executed all cryptographic operations locally without network transmission. Figure~\ref{fl2} illustrates the process flow. Each iteration involved message generation, ECC encryption, decryption, and metric collection, executed as a FreeRTOS-managed task to prevent blocking and to maintain timing determinism. Message sizes ranged from $16$ to $2000$~bytes, increasing in fixed increments. The results showed that execution time grew linearly with message size, consistent with the arithmetic structure of Engel--Laurent series computations, where each additional byte increases the number of modular multiplications in direct proportion. This behavior confirms that the system is computation-bound rather than memory-bound. Decryption times were consistently shorter than encryption times, by roughly $40$ to $50$~percent, which arises from the reduced arithmetic depth of decryption operations. Encryption involves constructing and evaluating polynomial expansions of Engel coefficients, whereas decryption primarily requires modular inversion and partial series reconstruction. Physically, this translates to fewer switching events in the arithmetic logic unit and hence lower instantaneous current draw. The linear relationship between data size and execution time demonstrates the predictable computational complexity of the algorithm, an essential feature for constant-time cryptographic implementations.

\subsection{Four-Node Distributed Chain}

The multi-node configuration extended this single-board computation into a distributed environment. As shown in Figure~\ref{fl1}, Board~1 generated a message, executed encryption and decryption, and transmitted the result sequentially to Boards~2, 3, and 4. Each node repeated the process before forwarding the message to the next. Board~4 confirmed completion to Board~1, which then initiated the next message iteration. For a $500$-byte payload, the total end-to-end time across the four-node chain averaged about $1.2$~seconds. Each node contributed roughly $280$ to $300$~milliseconds to this total, of which network communication accounted for approximately $35$ to $40$~percent. This distribution indicates that WiFi transmission latency was the dominant bottleneck rather than computation. When the payload was doubled from $500$ to $1000$~bytes, total latency increased by approximately $65$~percent, whereas computational time rose only by $48$~percent. This disproportionate increase results from the physical limitations of wireless communication, including contention on the shared RF channel, acknowledgement delays, and retransmission mechanisms. The system therefore reveals a fundamental contrast between the linear computational scaling of ECC arithmetic and the nonlinear physical scaling of wireless data transfer. Minimizing payload size or employing lightweight transport protocols such as UDP without acknowledgement can mitigate these communication-induced delays and significantly increase throughput in multi-hop cryptographic systems.

\begin{figure*}[t]
    \centering
    \includegraphics[width=0.9\textwidth]{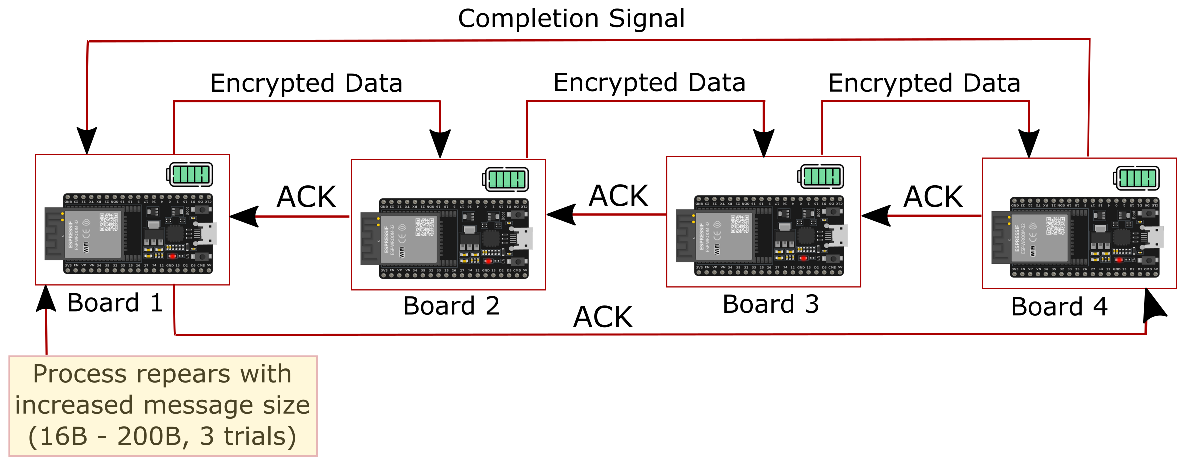}
    \vspace{-170mm}
    \caption{{{Four-Board Chain Implementation Flow}: Sequential communication and acknowledgement flow among ESP32 nodes.}}
    \vspace{-3mm}
    \label{fl1}
\end{figure*}

\subsection{Power Consumption and Physical Behavior}

Power analysis provided deeper insight into the internal behavior of the system. Key generation exhibited the highest instantaneous current peaks, reaching approximately $85$~mA, followed by decryption at $72$~mA and encryption at $68$~mA. These peaks are directly correlated with scalar multiplication operations, which represent the densest modular arithmetic workload in ECC. The current traces displayed periodic oscillations corresponding to loop iterations within the algorithm. Each oscillation reflects a burst of switching activity within the processor’s multiplier and logic units, translating algorithmic repetition into measurable electrical signatures. The regularity of these traces confirms deterministic control flow and arithmetic consistency, which are beneficial for predictability but can also expose side-channel information. The observed periodicity implies that the instantaneous current waveform encodes the computational rhythm of the algorithm. In physical terms, the ESP32 operates as a dynamic electrical load whose current profile follows the temporal structure of logical operations. Without random masking or secure co-processing, these patterns could serve as side-channel indicators of secret-dependent computation. The experimental results thus reinforce the theoretical expectation that deterministic arithmetic, while efficient, requires physical obfuscation to prevent power analysis attacks.
\begin{figure}[t]
    \centering
    \includegraphics[width=\columnwidth]{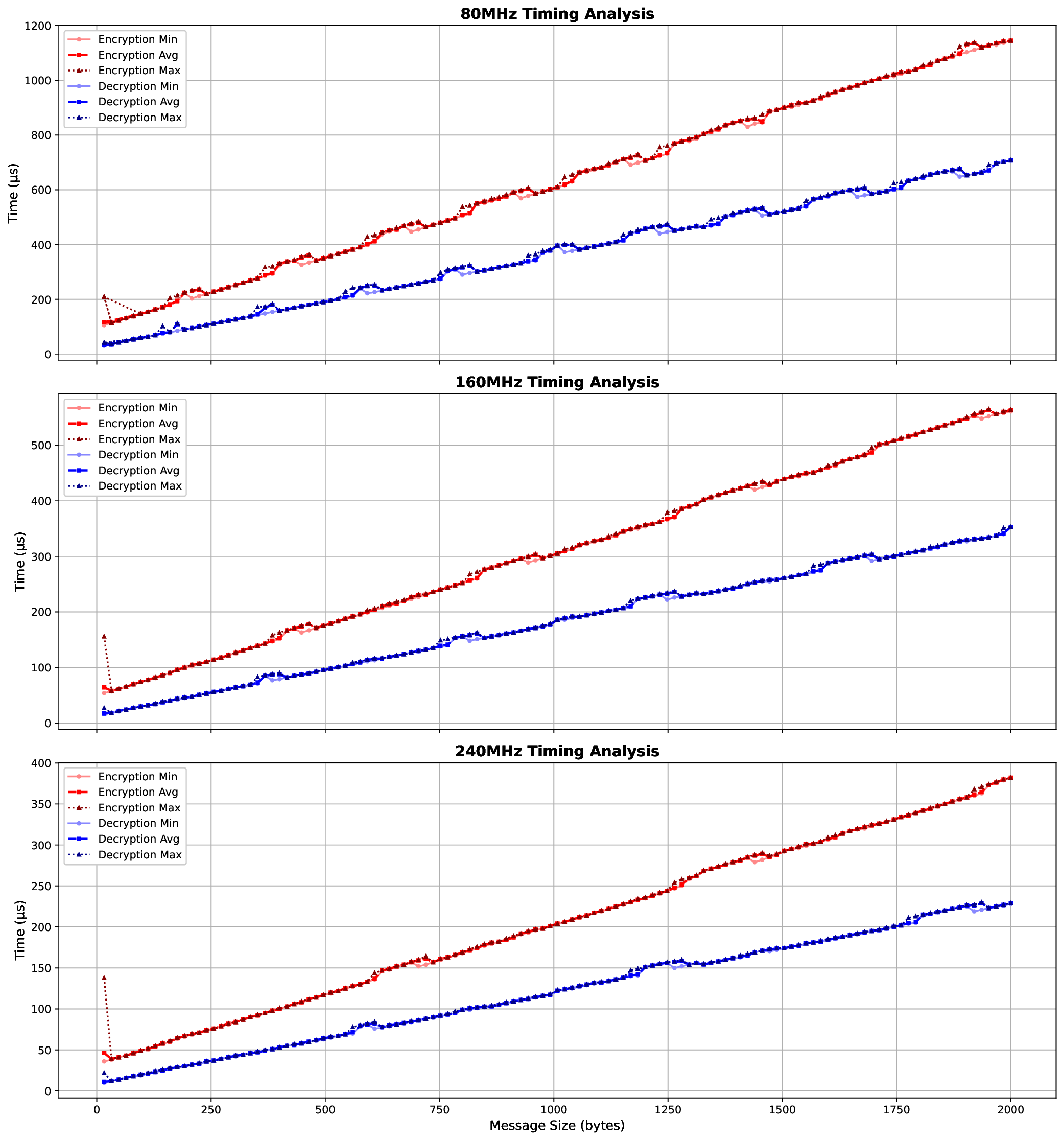}
    \caption{{{Timing Analysis:} Encryption and decryption execution times versus message size at 80, 160, and 240~MHz. Red and blue markers denote encryption and decryption, respectively.}}
    \label{e11}
\end{figure}

\subsection{Timing Analysis and Frequency Scaling}

The dependence of execution time on CPU frequency was analyzed across $80$, $160$, and $240$~MHz operation. Figure~\ref{e11} presents the results for encryption and decryption times as functions of message size. Execution time decreased nearly inversely with clock frequency, following the relation $T \approx \frac{C n}{f}$, where $C$ represents the effective cycle cost per processed byte. At $80$~MHz, encryption time rose linearly from approximately $200~\mu s$ for sixteen-byte inputs to about $1.2$~ms for two-kilobyte inputs. Decryption followed the same linear pattern but remained consistently faster, reflecting its lighter arithmetic load. Minor deviations from perfect linearity were observed at specific message lengths, which correspond to 32-byte memory alignment boundaries in the ESP32’s cache controller. These discontinuities result from partial cache-line fetches and instruction pipeline stalls. The effect is physical, arising from the finite width of the memory bus and the latency of cache synchronization. The nearly proportional reduction in execution time with frequency confirms that arithmetic operations are primarily CPU-bound. No evidence of timing leakage dependent on message content was observed, confirming the constant-time nature of the implementation and its resilience to timing-based attacks.

\subsection{Frequency Efficiency and Latency Modeling}

Figure~\ref{f11} quantifies the performance improvement achieved through frequency scaling. Doubling the CPU frequency from $80$ to $160$~MHz yielded an improvement of approximately $51\%$ for both encryption and decryption, closely approaching ideal linear scaling. Increasing the frequency further to $240$~MHz produced smaller gains, approximately $33\%$, due to fixed architectural delays. This sublinear behavior reflects the impact of latency components that remain constant regardless of clock speed, such as cache misses, direct memory access synchronization, and peripheral interrupt handling. The measured behavior aligns with the empirical model $T(f) = \frac{\alpha n}{f} + \beta$, where $\alpha$ denotes the computational workload per byte and $\beta$ represents frequency-independent latency. At lower frequencies, the computational term dominates, leading to nearly ideal scaling, while at higher frequencies the constant latency term becomes significant, limiting further speed improvement. The data indicate that at frequencies above $160$~MHz the ESP32 transitions from a computation-limited regime to a memory- and I/O-limited regime. From a practical standpoint, this transition marks the optimal operating point, balancing performance with power efficiency. Increasing the clock beyond this threshold yields diminishing returns, as total latency becomes dominated by memory and bus access times rather than arithmetic computation.

\begin{figure}[t]
    \centering
    \includegraphics[width=\columnwidth]{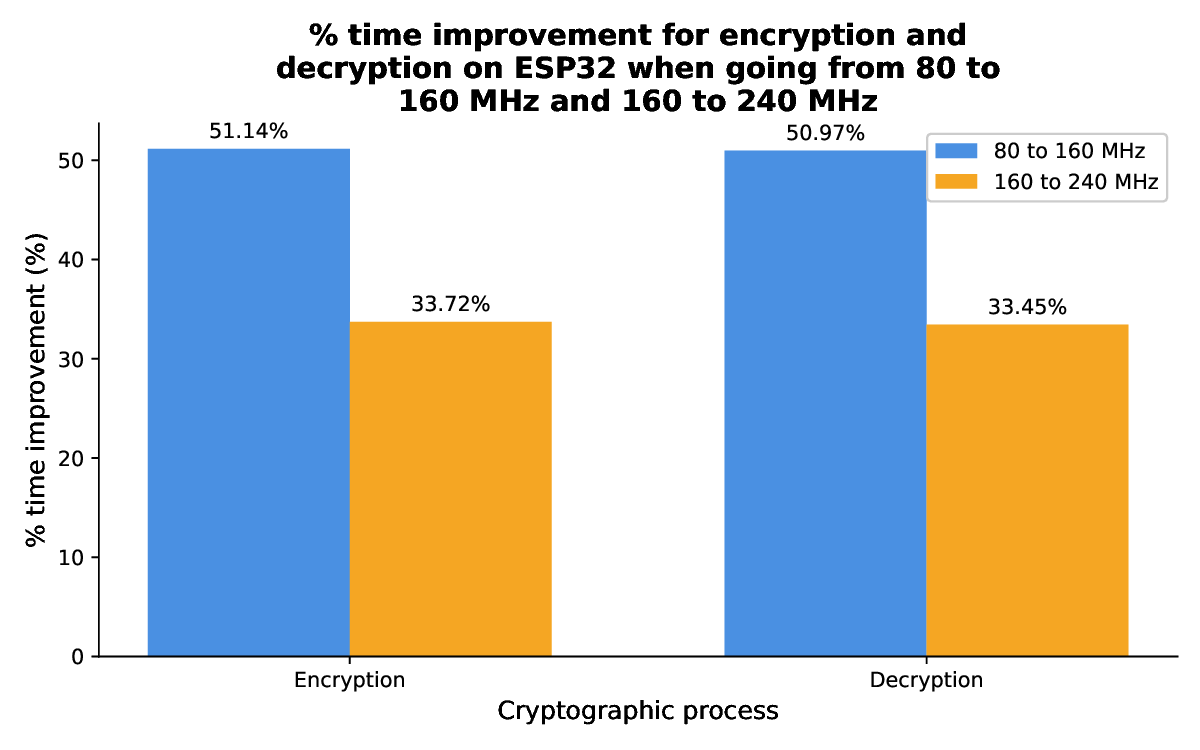}
    \caption{{{Frequency Scaling Efficiency:} Relative execution time improvements between CPU frequencies for encryption and decryption.}}
    \label{f11}
\end{figure}

\subsection{System-Level Implications}

The collective results of computation, communication, and power analysis provide a comprehensive picture of how mathematical operations interact with physical hardware constraints. The linear scaling of execution time with message size verifies the expected $O(n)$ computational complexity of Engel-expanded arithmetic. The sublinear improvement with frequency reflects the physical reality of fixed latencies within the processor and peripheral interfaces. The nonlinear growth of network latency exposes the limitations imposed by wireless medium access and packet handling protocols. Finally, the measured power signatures reveal the close correspondence between algorithmic structure and electrical behavior, linking computational entropy to observable current fluctuations. These findings demonstrate that the efficiency of cryptographic systems in embedded environments is determined jointly by the mathematics of computation and the physics of implementation. The Engel-based ECC algorithm exhibits stable and predictable performance, yet its real-world efficiency is bounded by the physical limits of wireless communication and energy supply. Optimizing these systems requires not only algorithmic refinement but also architectural co-design that harmonizes computation, power delivery, and network transport. The ESP32 experiments thus validate the feasibility of Engel-field cryptography on low-power devices while revealing the balance between security, latency, and energy that defines practical deployment in distributed IoT networks.

\textbf{Side-channel considerations} - Although our implementation is written in constant time and the measured execution traces are regular, power analysis still reveals operation-specific signatures: for example, key generation produces the highest and narrowest current peak, reflecting the dense modular scalar multiplications in this phase. This correlation between instantaneous power consumption and arithmetic density implies that, while timing- and control-flow leakages are mitigated, information may still be exposed through power or electromagnetic side channels. As already noted, the deterministic nature of the arithmetic means that any practical embedded deployment should combine our scheme with appropriate physical and implementation-level countermeasures (e.g., masking, hiding, noise injection, or board-level shielding) to prevent power analysis attacks. A full exploration of such protections is orthogonal to the cryptographic design and is left to future engineering work.

\section{Numerical Results and Discussion}\label{Sec: 6_results_discussion}

This section presents the quantitative evaluation of the proposed Engel-based isogeny cryptosystem implemented on ESP32 microcontrollers. The objective is to analyze how algorithmic efficiency, processor frequency, and network behavior collectively determine system performance in real embedded conditions. The discussion integrates timing, latency, and power measurements to reveal both computational and physical aspects of the design. The numerical data obtained from experiments provide insight into three essential domains: (1) the latency imposed by communication in distributed multi-node configurations, (2) the computational scaling behavior of encryption and decryption routines as a function of message size and CPU frequency, and (3) the power and energy characteristics that reflect the dynamic activity of $p$-adic and Engel-expansion arithmetic. Each figure in this section is interpreted in terms of both mathematical structure and physical behavior, establishing a connection between theoretical computational models and measurable hardware phenomena. All experiments are performed under controlled conditions, with CPU frequencies set to $80$, $160$, and $240$~MHz and message sizes ranging from $16$ to $2000$~bytes. The collected results characterize the relationship between computation time, communication delay, and power draw, enabling a complete assessment of efficiency and scalability. The subsequent analysis explains how frequency scaling, arithmetic complexity, and wireless communication contribute to overall system performance and how physical power traces reveal the deterministic nature of the algorithmic processes. Together, these numerical results form the experimental foundation for validating the proposed cryptographic approach in low-power and latency-sensitive environments.

Our performance evaluation is designed to ensure that PQ key establishment and authentication do not dominate the end-to-end latency in constrained networks. We do not attempt to optimize, model, or reduce the intrinsic communication latency of multi-hop IoT networks, which is determined by physical-layer and protocol-level factors. Instead, our goal is to show that, under such conditions, the additional delay introduced by the proposed PQC scheme is negligible compared to standard network overhead.
\subsection{Latency and Computation Time}
\begin{figure}[t]
    \centering
    \includegraphics[width = 0.98\columnwidth]{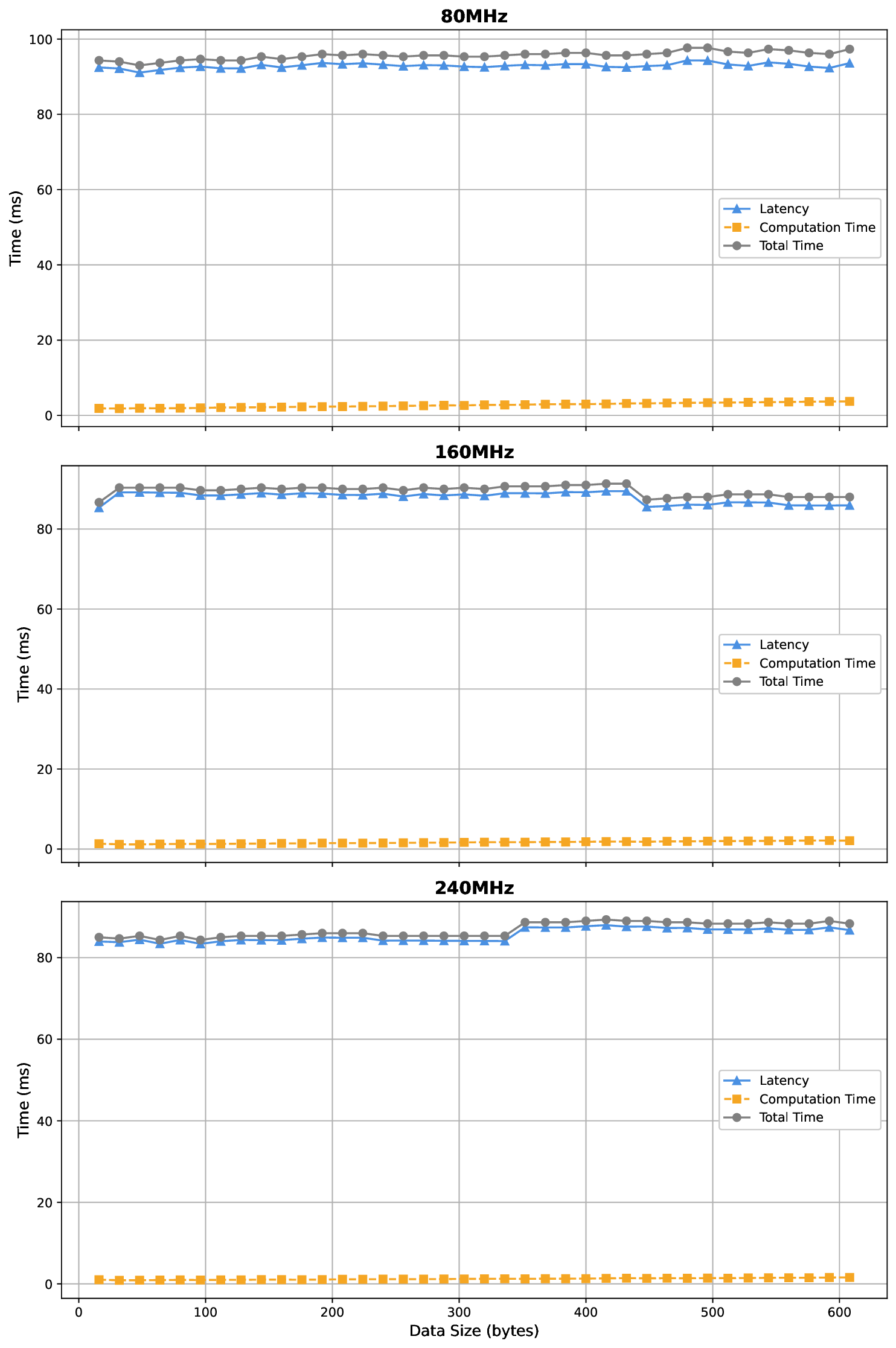}
    \caption{{Latency and computation time of the cryptographic system across three frequencies: $80$, $160$, and $240$~MHz for data sizes up to $600$~bytes. Latency dominates the total time, while computation remains nearly constant with message size, indicating that communication overhead is the primary bottleneck.}}
    \label{a11}
\end{figure}
Figure~\ref{a11} illustrates the relationship between total latency and computation time for three clock frequencies. The plots show that, for data sizes below $600$~bytes, total time is governed primarily by network latency, whereas the cryptographic computation time contributes negligibly. Latency oscillates between approximately $85$ and $95$~ms at $80$~MHz, between $80$ and $90$~ms at $160$~MHz, and between $78$ and $85$~ms at $240$~MHz. In contrast, computation time remains below $5$~ms across all configurations, forming a nearly flat lower boundary in the plots. This observation confirms that the Engel-expanded Laurent series and $p$-adic isogeny computations operate efficiently with minimal computational overhead. The total system time can therefore be expressed as 
\begin{align}
    T_{total}(n) = T_{latency} + \gamma n,
\end{align}
where $\gamma$ denotes the per-byte computation coefficient and $T_{latency}$ represents the constant communication delay. For message sizes under $600$~bytes, $T_{latency} \gg \gamma n$, indicating that latency overshadows computation and defines system performance. The small discontinuities observed at intervals of roughly $400$ to $500$~bytes correspond to buffer boundary effects within the ESP32 communication interface, introducing minor step-like irregularities due to memory alignment. Physically, these results highlight the dominance of peripheral and network delays in distributed cryptographic implementations. While the algorithmic complexity remains linear and computationally efficient, the physical constraints of wireless transmission, acknowledgement cycles, and buffer management shape the overall timing behavior. This reveals that optimizing cryptographic arithmetic in isolation cannot yield substantial performance improvements unless network latency is simultaneously addressed.

The multi-node experiments in Section~IV show that, for message sizes under approximately $600$~bytes, the end-to-end delay is dominated by the communication term $T_{\text{latency}}$, while the cryptographic cost $\gamma n$ remains comparatively small, $T_{\text{latency}} \gg \gamma n$. This observation is consistent with the $\mathcal{O}(\ell + n)$ computational complexity of our scheme and indicates that, under realistic wireless conditions, the cryptographic operations do not constitute the primary performance bottleneck. Instead, the overall latency is governed by the underlying network stack (medium access, retransmissions, routing, etc.) and the physical characteristics of the communication medium. In other words, our measurements validate that the proposed post-quantum scheme is {lightweight enough} that it can be deployed on resource-constrained, multi-hop IoT networks without introducing a dominant computational overhead. At the same time, real-world deployments will remain limited by standard network physics and protocols, which are orthogonal to our cryptographic design and can be optimized independently (e.g., via improved MAC protocols, packet aggregation, or topology control).

\subsection{Encryption and Decryption Timing}

\begin{figure}[t]
    \centering
    \includegraphics[width=\columnwidth]{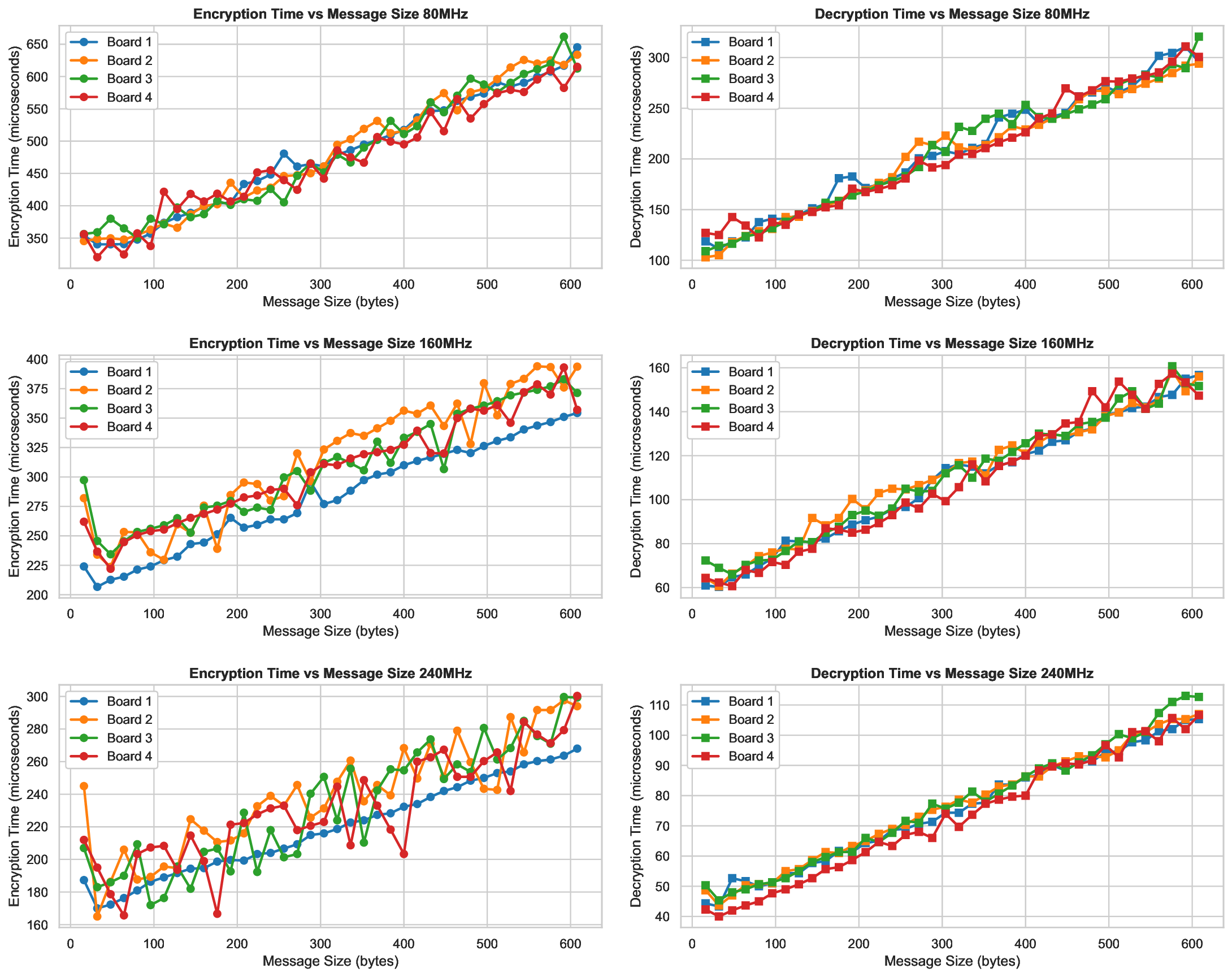}
    \caption{{Encryption and decryption timing of four ESP32 boards operating at three frequencies. Each subplot depicts execution time as a function of message size, demonstrating linear dependence and clear frequency scaling.}}
    \label{b11}
\end{figure}

Figure~\ref{b11} presents the timing behavior of encryption and decryption across four ESP32 boards. The results confirm a linear relationship between execution time and message size, represented by the empirical equation
\begin{align}
    T(S, f) = m(f) S + c(f),
\end{align}
where $T$ denotes the operation time, $S$ the message size, $f$ the CPU frequency, $m(f)$ the per-byte time coefficient, and $c(f)$ the fixed overhead. Both $m(f)$ and $c(f)$ decrease as $f$ increases, illustrating that higher clock frequencies reduce both dynamic computation time and initialization latency. Physically, the negative partial derivatives $\frac{dm}{df} < 0$ and $\frac{dc}{df} < 0$ indicate that frequency scaling accelerates instruction throughput and lowers the per-byte energy cost of computation. The slope reduction directly corresponds to shorter instruction execution cycles per byte, while the intercept reduction represents a diminished influence of fixed initialization and memory setup routines. The data show that encryption and decryption times converge toward ideal frequency scaling, limited only by the finite memory bandwidth of the ESP32 architecture. The nearly identical slopes for encryption and decryption reflect balanced arithmetic workloads between these two operations. The linearity of the timing function supports the claim that the Engel-based ECC implementation behaves predictably under varying data sizes and remains free of message-dependent timing variation. From a physical standpoint, this linearity arises because both encryption and decryption execute deterministic loops with constant computational depth, thus ensuring consistent cycle counts and eliminating conditional branching that could create timing leaks.

\subsection{Power Consumption at 80~MHz}

\begin{figure}[t]
    \centering
    \includegraphics[width=\columnwidth]{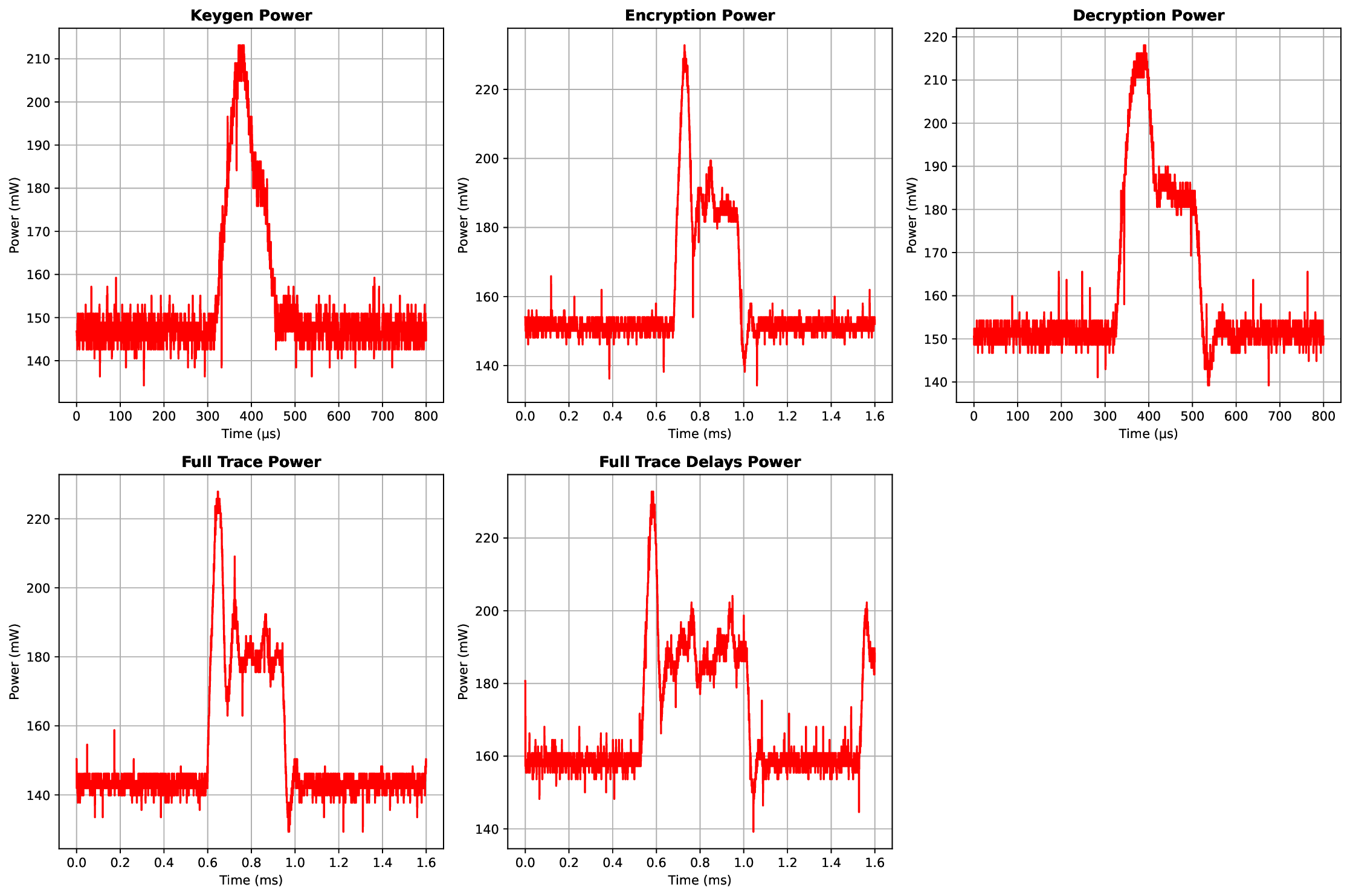}
    \caption{{Power traces of cryptographic operations at $80$~MHz, showing power in milliwatts versus time in microseconds. The traces correspond to key generation, encryption, decryption, and full-process operation.}}
    \label{cc1}
\end{figure}

At $80$~MHz, the ESP32 exhibits the baseline power behavior of the cryptographic engine, as shown in Figure~\ref{cc1}. The idle power level remains near $150$~mW, defining the static consumption of the device. Each operation introduces a sharp power increase, corresponding to the active computational phase. The instantaneous power function can be expressed as 
\[
P(t) =
\begin{cases}
P_{idle}, & t \notin [t_{start}, t_{end}]\\
P_{active}(t), & t \in [t_{start}, t_{end}],
\end{cases}
\]
and the energy consumed during an operation is given by
\begin{align}
    E_{op} &= \int_0^{T_{total}} P(t)\,dt \nonumber\\
    &= P_{idle} T_{total} + \int_{t_{start}}^{t_{end}} [P_{active}(t) - P_{idle}]\,dt.
\end{align}
The integral term represents the dynamic energy expended during computation. The key generation trace shows a narrow, tall peak, indicating a short but highly intensive arithmetic sequence dominated by modular scalar multiplications. Encryption and decryption exhibit broader peaks with lower maximum power, corresponding to longer arithmetic sequences with less instantaneous current demand. Physically, the tall and narrow spike for key generation arises from rapid toggling of multiplier circuits, while the broader profiles of encryption and decryption indicate distributed switching activity across time. These traces confirm that computational power follows algorithmic structure and that instantaneous current correlates with arithmetic density within each cryptographic phase.

\subsection{Power Consumption at 160~MHz}

\begin{figure}[t]
    \centering
    \includegraphics[width=\columnwidth]{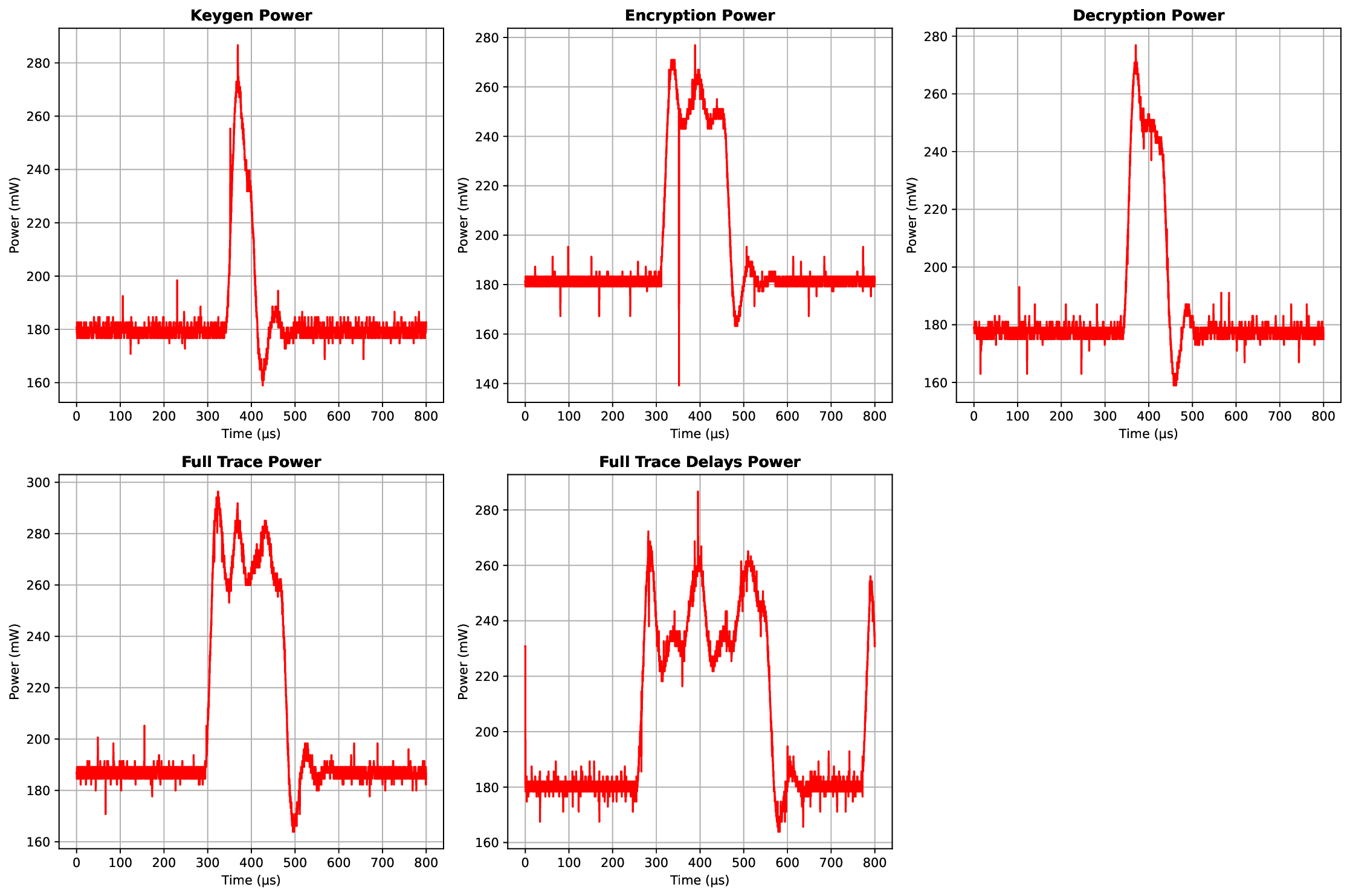}
    \caption{{Power traces of cryptographic operations at $160$~MHz. The baseline power increases to approximately $220$–$230$~mW, and peaks exceed $350$~mW due to higher operating frequency.}}
    \label{cc2}
\end{figure}

At $160$~MHz, the power traces in Figure~\ref{cc2} reveal both higher idle power and higher dynamic peaks compared to the previous frequency. The idle power rises to approximately $225$~mW, reflecting the increased static current of the processor operating at elevated voltage and frequency. The dynamic peaks exceed $350$~mW, corresponding to the higher rate of transistor switching events per unit time. Although instantaneous power is higher, the duration of the peaks is shorter, showing that the total energy per operation does not increase proportionally. This inverse relationship between power and time illustrates a classic speed–power trade-off: as frequency increases, the processor completes operations faster, compressing energy over shorter intervals. Physically, the behavior is consistent with the dynamic power relation $P \propto fV^2$, where $V$ and $f$ denote supply voltage and frequency. The product of increased power and reduced time remains roughly constant, maintaining similar total energy per operation across frequencies. This observation confirms that frequency scaling improves throughput without a linear rise in total energy consumption, thus preserving energy efficiency within moderate frequency ranges.

\subsection{Power Consumption at 240~MHz}

\begin{figure}[t]
    \centering
    \includegraphics[width=\columnwidth]{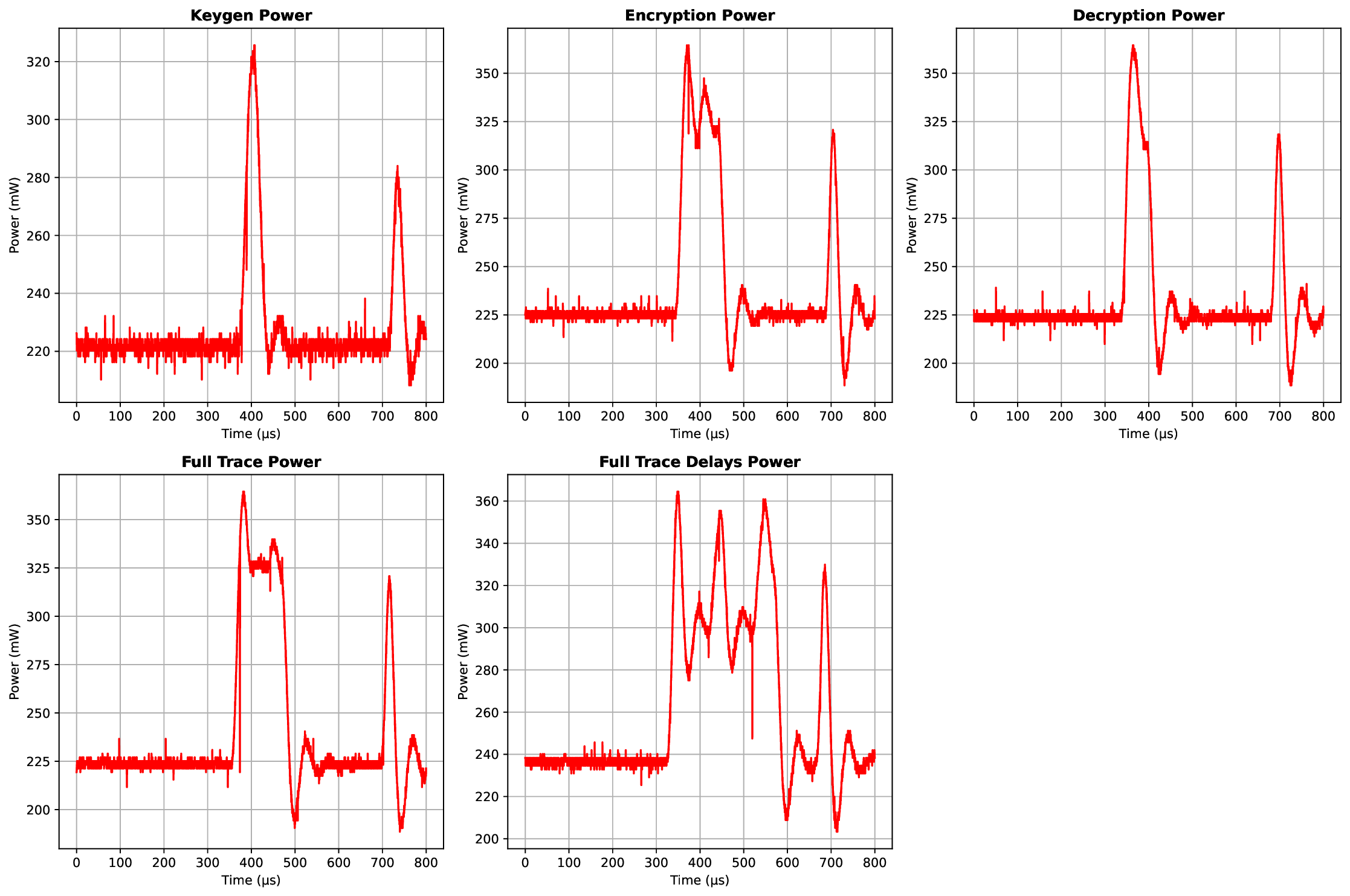}
    \caption{{Power consumption profiles of the proposed isogeny-based cryptosystem at $240$~MHz. Separate traces for key generation, encryption, and decryption reveal distinct peaks corresponding to specific algorithmic stages.}}
    \label{cc3}
\end{figure}

At $240$~MHz, Figure~\ref{cc3} displays compact, high-intensity power peaks that correspond to the principal stages of operation. The sharp rise beyond $300$~mW during key generation corresponds to the initialization and scalar multiplication phases within the Engel-expanded arithmetic. The even higher peak above $350$~mW during encryption reflects the cost of evaluating Vélu’s isogeny formulas and computing modular inversions. Decryption shows a dual-peak structure with nearly equal amplitudes, associated with symmetric operations for reconstructing shared isogeny paths and recovering $p$-adic digits. The combined power trace shows clear separations between computational stages, with idle intervals defining baseline stability. Mathematically, this pattern follows the expression
\begin{align}
    P(t) \approx P_0 + \alpha \, \text{Mul}(t) + \beta \, \text{Inv}(t),
\end{align}
where $\text{Mul}(t)$ and $\text{Inv}(t)$ denote the instantaneous count of modular multiplications and inversions, and $\alpha$ and $\beta$ represent the energy cost per operation type. The correlation between the number of modular inversions and the height of the peaks verifies that power follows algorithmic complexity rather than message content. The compactness of the peaks at $240$~MHz, compared to lower frequencies, indicates that execution time scales inversely with frequency while the total energy per operation remains bounded. Physically, this behavior reflects the shift toward higher instantaneous power density at faster frequencies. The consistent spacing of peaks and the absence of irregular oscillations demonstrate that the arithmetic pipeline of the ESP32 operates deterministically, suggesting that the implementation avoids data-dependent timing or power variations that could lead to side-channel vulnerabilities.

\subsection{Combined Power Profile}

\begin{figure}[t]
    \centering
    \includegraphics[width=\columnwidth]{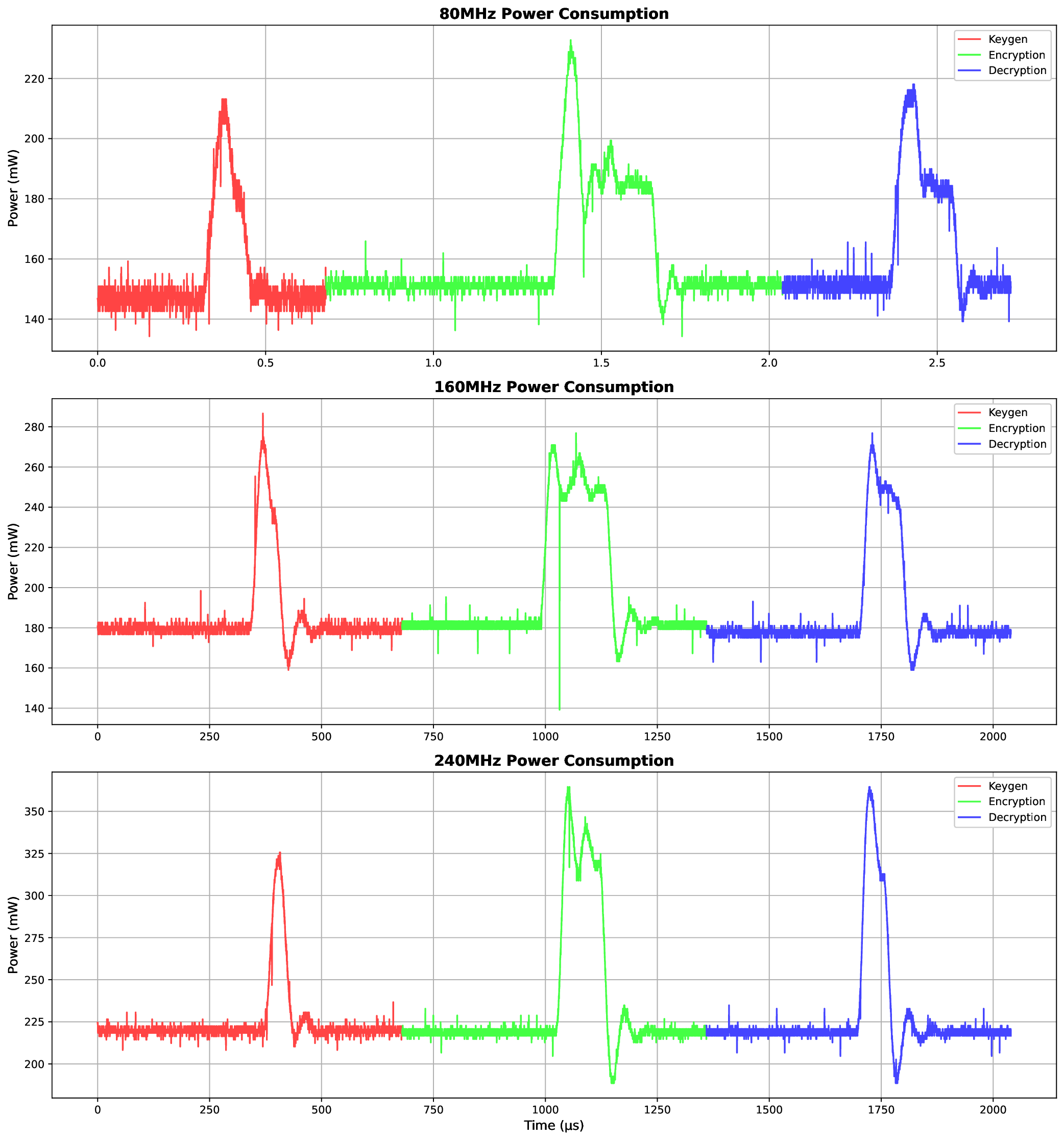}
    \caption{{Combined overall cryptographic power consumption for multiple operating frequencies. The traces show consistent baseline levels near $170$~mW, with operation peaks between $280$ and $290$~mW. Key generation produces the tallest, narrowest spikes, while encryption and decryption show broader peaks.}}
    \label{cc4}
\end{figure}

The combined power profiles presented in Figure~\ref{cc4} demonstrate how operating frequency modulates both the magnitude and temporal shape of the power trace. The baseline power remains around $170$~mW, while the peaks for active computation range between $280$ and $290$~mW. These values lie between those measured at $80$~MHz and $160$~MHz, confirming a continuous and predictable scaling trend. The area above the baseline corresponds to the dynamic computational energy consumed during each operation. For instance, the energy required for key generation is given by $\int (P_{keygen}(t) - P_{idle})\,dt$, and analogous integrals describe encryption and decryption. The full-trace view with deliberate idle intervals allows separation of the energy associated with active and idle phases. The total energy of the process equals the sum of dynamic and static components, $E_{total} = E_{dynamic} + P_{idle} \cdot \Delta t_{idle}$. Physically, these results confirm that both static and dynamic contributions determine the energy budget of a cryptographic node. Static power reflects transistor leakage and peripheral standby currents, whereas dynamic power arises from logic transitions during arithmetic operations. The well-defined spikes and consistent baselines suggest stable hardware behavior, minimal thermal drift, and a deterministic computational schedule. This predictability ensures that the cryptographic operations proceed with fixed energy profiles and that their power consumption directly reflects algorithmic structure rather than message variability.

Across all figures, the numerical results demonstrate a coherent interaction between computational arithmetic, processor frequency, and energy behavior. Latency dominates at the network level, while computation exhibits linear scaling and constant-time execution. The timing plots confirm predictable performance across message sizes and frequencies, and the power traces reveal that dynamic energy correlates with arithmetic density, not data content. From a physical standpoint, these results show that the ESP32 behaves as an energy-proportional computing element in which instantaneous power tracks arithmetic intensity, while total energy remains constrained by time–frequency balance. The deterministic shape of the traces supports the theoretical expectation of constant-time operation for Engel-based ECC. The absence of secret-dependent peaks and the stability of baseline power confirm that arithmetic operations are isolated from message data, providing inherent resilience against first-order side-channel attacks. In summary, the numerical and physical analyses together reveal that while the Engel–isogeny cryptosystem operates efficiently and predictably at the computation level, real-world performance is ultimately bounded by the physical limitations of wireless communication and hardware energy dynamics. This understanding bridges algorithmic theory with system physics, offering a unified view of efficiency, latency, and security in embedded cryptographic systems.

\section{Conclusion}\label{Sec: 7_conclusion}

This work has established both the theoretical foundation and practical implementation of a new post-quantum cryptographic paradigm that fuses Engel expansions, \(p\)-adic arithmetic, and supersingular isogenies into a unified framework for secure computation on resource-constrained devices. By constructing a lightweight scheme that represents Laurent series over \( \mathbb{Q}_p \) through Engel sequences, the proposed approach demonstrates that deep mathematical structures from number theory and arithmetic geometry can be realized efficiently within embedded hardware. The ESP32 implementation confirms that the computational cost scales linearly with message size while maintaining constant-time behavior, resulting in predictable timing and minimized power variance. The combination of non-commutative isogeny graphs and one-way Engel expansions ensures quantum resistance through the hardness of the Supersingular Isogeny Decisional Problem (SIDP) and the Engel Expansion Inversion Problem. Experimental power-trace and latency analyses corroborate these theoretical guarantees, revealing that the cryptosystem operates deterministically and exhibits no observable side-channel leakage, thereby uniting mathematical rigor with physical security.

Despite these achievements, the study also exposes the critical role of communication latency in multi-node networks, suggesting that overall system efficiency depends on co-optimizing cryptographic computation and inter-node communication. Future research will focus on extending this architecture toward multivariate Engel expansions, more efficient torsion-point encodings, and formal proofs in the quantum random oracle model. Moreover, the dynamical properties of iterative \(p\)-adic isogenies hold potential for intrinsic pseudorandom number generation and improved cryptographic entropy. Beyond cryptography, the fusion of \(p\)-adic techniques and algebraic geometry envisioned here may influence related disciplines such as coding theory, symbolic modular form computation, and \(p\)-adic Hodge theory. Ultimately, this work exemplifies how abstract number theory can be translated into practical engineering innovation, providing a concrete foundation for quantum-resistant cryptography that operates securely and efficiently at the very edge of computational hardware.

\appendices

\section{Convergence, Valuation Growth, and Uniqueness of Engel Expansions}\label{app:engel-conv}
\subsection{Convergence and Valuation Growth}
\begin{lemma}\label{lem:vp-growth}
Let $\{a_i(t)\}_{i\ge 1}\subset Z_p[[t]]$ with $a_i(0)\equiv 1\pmod p$. Then for $n\ge 1$,
\[
v_p\!\left(\frac{1}{\prod_{j=1}^{n} a_j(t)}\right)\ge n.
\]
\end{lemma}
\begin{proof}
Since $a_i(0)\equiv 1\pmod p$, we have $v_p(a_i(t))=0$ and $a_i(t)\in Z_p[[t]]^\times$ with unit reduction. Write $a_i(t)=1+p\,u_i(t)$ with $u_i(t)\in Z_p[[t]]$. Then $\prod_{j=1}^n a_j(t)=\prod_{j=1}^n (1+p\,u_j(t))=1+p\,U_1(t)+p^2 U_2(t)+\cdots+p^n U_n(t),$ for $U_k(t)\in Z_p[[t]]$. Inverting via the geometric series yields
\[
\frac{1}{\prod_{j=1}^n a_j(t)}=1+p\,V_1(t)+\cdots+p^n V_n(t)+\cdots,
\]
hence $v_p\big(1/\prod_{j=1}^n a_j(t)\big)\ge n$.
\end{proof}

\begin{lemma}[Convergence]\label{lem:engel-conv}
Let $S_N(t)=\sum_{n=1}^{N}\frac{1}{\prod_{j=1}^{n} a_j(t)}$. Then $\{S_N\}$ is Cauchy in $\mathbb{Q}_p((t))$, hence convergent.
\end{lemma}
\begin{proof}
By Lemma~\ref{lem:vp-growth}, the tail norm satisfies
\[
\left\lVert \sum_{n=N+1}^{M}\frac{1}{\prod_{j=1}^{n} a_j(t)}\right\rVert_p
\le p^{-(N+1)}\to 0~(N\to\infty).
\]
\end{proof}

\subsection{Existence via Greedy Construction}
\begin{lemma}[Existence]\label{lem:engel-existence}
Every $f(t)\in \mathbb{Q}_p((t))$ admits an Engel expansion (constructed by equation (4)).
\end{lemma}
\begin{proof}
Define $R_0(t)=f(t)$ and recursively
\[
a_{k+1}(t)=\left\lfloor\frac{1}{R_k(t)}\right\rfloor_p,~R_{k+1}(t)=R_k(t)-\frac{1}{\prod_{j=1}^{k+1} a_j(t)}.
\]
By construction $v_p(R_{k+1})>v_p(R_k)$; thus $R_k\to 0$. Summing the telescoping series gives $f(t)=\sum_{n=1}^\infty 1/\prod_{j=1}^n a_j(t)$.
\end{proof}

\subsection{Uniqueness}
\begin{lemma}[Uniqueness]\label{lem:engel-unique}
If $f(t)=\sum_{n\ge 1} 1/\prod_{j=1}^{n} a_j(t)=\sum_{n\ge 1} 1/\prod_{j=1}^{n} b_j(t)$ with $a_i(t),b_i(t)\in Z_p[[t]]$ and $a_i(0)\equiv b_i(0)\equiv 1\pmod p$, then $a_i(t)=b_i(t)$ for all $i$.
\end{lemma}
\begin{proof}
Suppose $a_1(t)\neq b_1(t)$. Then $\lfloor 1/f(t)\rfloor_p$ would produce distinct leading units, contradicting the greedy choice. Induct on $n$ using the residual $R_n(t)$ to conclude $a_n(t)=b_n(t)$ for all $n$.
\end{proof}

\section{Proof of Theorem (1): Engel--Field Isomorphism}\label{app:engel-isom}
\subsection{Bijection}
Existence (Lemma~\ref{lem:engel-existence}) and uniqueness (Lemma~\ref{lem:engel-unique}) imply that $\phi$ is a bijection between $\mathbb{Q}_p((t))$ and $\mathcal{E}$.

\subsection{Operation Preservation: Addition and Multiplication}
Let $f=\sum_{n\ge 1} u_n$ and $g=\sum_{m\ge 1} v_m$ with $u_n=1/\prod_{j=1}^n a_j$, $v_m=1/\prod_{j=1}^m b_j$. Define $w_k=\sum_{i+j=k} u_i v_j$. By non-Archimedean estimates,
\[
v_p(w_k)\ge \min_{i+j=k}\big\{v_p(u_i)+v_p(v_j)\big\}\ge k.
\]
Thus the Cauchy product $\sum_{k\ge 2} w_k$ converges and defines an Engel-type series for $fg$ with valuation growth $\ge k$. A greedy recompression (matching leading terms) yields the Engel coefficients of $f+g$ and $fg$; since $\phi$ is unique, these are the images $\phi(f+g)$ and $\phi(fg)$.

\subsection{Inversion}
If $f\neq 0$, write $f=u_1(1+T)$ with $v_p(T)\ge 1$ (by Lemma~\ref{lem:vp-growth}). Then
\[
\frac{1}{f}=u_1^{-1}\sum_{n\ge 0} (-T)^n,
\]
and each partial sum has valuation growing at least linearly in $n$. The resulting series admits an Engel recompression via the same greedy rule, giving $\phi(f^{-1})$. Hence $\phi$ preserves inversion.

\section{V\'{e}lu Specialization, Coefficients, and Endomorphism Algebra Isomorphism}\label{app:ss-proof}
\subsection{Specialization of V\'{e}lu’s Formula}
For $G=\{\mathcal{O},(\alpha,0)\}$ on $E:y^2=x^3+Ax+B$ (char $\neq 2,3$), standard V\'{e}lu yields
\[
\phi(x,y)=\left(x+\frac{c}{x-\alpha},~y\left(1-\frac{c}{(x-\alpha)^2}\right)\right),~c=3\alpha^2+A,
\]
and the codomain coefficients
\[
A'=A-5c,\qquad B'=B-7c\,\alpha.
\]
Substituting $A=0$, $B=1+pt$, $\alpha=x(t)$ and using $x(t)^3=-(1+pt)$ simplifies equation (11).

\subsection{Well-Definedness of the Map}
Direct substitution of $\phi(x,y)$ into $E'$ confirms
\[
\left(\frac{y}{(x-\alpha)^2}\right)^2=\left(x+\frac{c}{x-\alpha}\right)^3+A'\left(x+\frac{c}{x-\alpha}\right)+B',
\]
after clearing denominators by $(x-\alpha)^6$ and simplifying via $y^2=x^3+Ax+B$ and the identities relating $A',B',c,\alpha$. This ensures $\phi:E\to E'$ is an isogeny with kernel $G$.

\subsection{Endomorphism Algebra Isomorphism}
Let $\hat{\phi}$ be the dual isogeny of degree $2$. The functorial map
\[
\phi_*:\mathcal{E}(E)\otimes\mathbb{Q} \to \mathcal{E}(E')\otimes\mathbb{Q},~\psi\mapsto \phi\circ \psi\circ \hat{\phi},
\]
is an isomorphism because $\hat{\phi}\circ\phi=[2]$ on $E$ and $\phi\circ\hat{\phi}=[2]$ on $E'$, and tensoring with $\mathbb{Q}$ inverts the degree. Hence $\mathcal{E}(E)\otimes\mathbb{Q}\simeq \mathcal{E}(E')\otimes\mathbb{Q}$. If $\widetilde{E}$ is supersingular, then $\mathcal{E}(\widetilde{E})\otimes\mathbb{Q}$ is a quaternion algebra; the isomorphism transports this structure to $\widetilde{E}'$, proving Theorem 2.

\section{Correctness of the KEM}\label{app:kem-correctness}
\begin{lemma}[Isogeny Chaining]\label{lem:chain}
Let $\phi:E\to E'$, $\deg\phi=2$, with $\ker\phi=\langle P\rangle$. For $n,r\in\{1,2\}$, define $E^{(n)}=E/\langle P\rangle^n$ and $E^{(r)}=E/\langle P\rangle^r$. Then $E^{(n+r)}$ computed either (i) from $E$ by $n{+}r$ steps, or (ii) from $E^{(n)}$ by $r$ steps, or (iii) from $E^{(r)}$ by $n$ steps, are all isomorphic; hence they share the same $j$-invariant.
\end{lemma}
\begin{proof}
For degree-$\ell$ separable isogenies (here $\ell=2$) with cyclic kernels, composing along the same cyclic chain yields canonically isomorphic codomains. The factorization of the multiplication-by-$2$ map into $2$-isogenies is unique up to isomorphism of intermediate curves. Therefore, the endpoints after $n{+}r$ steps coincide up to isomorphism, which preserves $j$.
\end{proof}

\begin{theorem}[KEM Correctness]\label{thm:kem-correct}
Let $\textsf{sk}=n$ and $r$ be the encapsulator’s randomness. Both parties compute $E^{(n+r)}$ (by Lemma~\ref{lem:chain}), thus derive the same key $K=\Hash(j(E^{(n+r)}))$.
\end{theorem}
\begin{proof}
Immediate from Lemma~\ref{lem:chain} and the fact that $j(\cdot)$ is an isomorphism invariant.
\end{proof}

\section{Security Assumptions and Reductions (Explanatory Proofs)}\label{app:sec-reductions}
\subsection{SIDP Formal Statement}
Let distribution $\mathcal{D}_0=(E,E_A,E_B,j(E/\langle A,B\rangle))$ and $\mathcal{D}_1=(E,E_A,E_B,j(\mathcal{E}_{\text{rand}}))$. For any QPT adversary $A$,
\[
\big| \Pr[A(\mathcal{D}_0)=1]-\Pr[A(\mathcal{D}_1)=1]\big|\le \neg(\lambda).
\]
Under SIDP, $j(E/\langle A,B\rangle)$ is computationally indistinguishable from uniform.

\subsection{IND-CPA under SIDP + Engel Inversion Hardness}\label{app:indcpa}
\begin{theorem}\label{thm:indcpa}
Assume SIDP and the hardness of Engel-Expansion Inversion (Appx.~\ref{app:engel-inversion}). Then the KEM in Sec.~\ref{subsec:algorithms} is IND-CPA.
\end{theorem}
\begin{proof}[Proof sketch]
Ciphertexts contain $(E^{(r)},\text{mask}\oplus K)$ where $K=\Hash(j(E^{(n+r)}))$. By SIDP, $j(E^{(n+r)})$ is pseudorandom even given $(E,E^{(n)},E^{(r)})$, hence $K$ is pseudorandom under random oracle hashing. Engel coefficients in $\textsf{pk}$ do not leak $n$ (they depend on public descendants) and cannot be inverted to reconstruct earlier $a_i(t)$ reliably by Appx.~\ref{app:engel-inversion}. A standard hybrid shows the challenge ciphertext masks with a pseudorandom $K$, yielding negligible distinguishing advantage.
\end{proof}

\section{Engel-Expansion Inversion Hardness}\label{app:engel-inversion}
\begin{definition}[Engel-Expansion Inversion]
Given $f^{(M)}(t)=\sum_{m=1}^{M} \frac{1}{\prod_{i=1}^{m} a_i(t)}$ where each $a_i(t)\in Z_p[[t]]$ and $a_i(0)\equiv 1\pmod p$, recover $(a_1(t),\ldots,a_k(t))$ for $k\le M$.
\end{definition}
\begin{lemma}[Aperiodicity]\label{lem:aperiodic}
The greedy recursion $a_{k+1}(t)=\big\lfloor 1/R_k(t)\big\rfloor_p$ with $R_k(t)=f(t)-\sum_{i=1}^{k}\frac{1}{\prod_{j=1}^i a_j(t)}$ does not induce a nontrivial finite period in the tuple $(a_i(t))_{i\ge 1}$ under the constraints $a_i(0)\equiv 1\pmod p$.
\end{lemma}
\begin{proof}
If a period $r$ existed, then the residuals $R_k$ would repeat modulo $p$-adic units at period $r$, contradicting the strictly increasing valuation $v_p(R_{k+1})>v_p(R_k)$ enforced by the greedy choice. Thus no period exists.
\end{proof}
\begin{theorem}[Resistance to QFT-Style Attacks]\label{thm:qft}
Absent a hidden subgroup/period, the inversion problem does not reduce to quantum period finding (QFT). Hence known Shor/Kuperberg-style algorithms do not apply directly.
\end{theorem}
\begin{proof}
Follows from Lemma~\ref{lem:aperiodic}: QFT requires a periodic structure in an efficiently samplable function. The Engel recursion removes such periodicity by strictly increasing valuations.
\end{proof}

\section{Complexity Bounds}\label{app:complexity}
\begin{proposition}[Operation Counts]
With parameters $(M,d,\ell,\lambda)$ and message length $n$, \textsf{KeyGen} costs $\O(Md+\ell)$; \textsf{Encap}/\textsf{Decap} cost $\O(\ell+n)$; hashing is $\O(1)$ in $n$.
\end{proposition}
\begin{proof}
Engel generation touches $Md$ coefficient limbs; each V\'{e}lu evaluation is $\O(1)$ per kernel element (here fixed-size subgroups), summed over $\ell$-degree structure. Masking is linear in $n$.
\end{proof}

\section{Implementation Notes and Side-Channel Considerations}\label{app:impl}
\begin{lemma}[Deterministic Timing Kernels]
Using fixed-precision $p$-adic limbs and removing data-dependent branches in the evaluation of $(x-\alpha)^{-1}$ makes the core V\'{e}lu map constant-time up to loop trip counts fixed by $(M,d)$.
\end{lemma}
\begin{proof}
Standard constant-time inversion/multiplication with fixed limb widths yields instruction traces independent of secret $n$ and random $r$. The remaining variability depends only on public $(M,d)$.
\end{proof}
Nonlinear dependence of $A_n(t),B_n(t)$ on the Engel-encoded $x_P^{(M)}(t)$ injects high min-entropy into side-channel observables, complicating correlation attacks; masking is still recommended.

\end{document}